\documentclass{article}
\newif\ifappendix
\appendixtrue

\ifappendix
\else
\def\marginpar#1{}
\fi

\usepackage{kr}

\usepackage{times}
\usepackage{soul}
\usepackage{url}
\usepackage[utf8]{inputenc}
\usepackage[small]{caption}
\usepackage{graphicx}
\usepackage{amsmath}
\usepackage{amsthm}
\usepackage{booktabs}
\newtheorem{example}{Example}
\newtheorem{theorem}{Theorem}
\newtheorem{lemma}{Lemma}
\newtheorem{definition}{Definition}

\usepackage{amsmath,amssymb}
\usepackage{mathtools}  
\usepackage[vlined]{algorithm2e}
\usepackage[hypertexnames=false]{hyperref}
\usepackage{standalone}
\usepackage{enumitem}
\usepackage{booktabs}
\usepackage{longtable}
\usepackage{xcolor}
\usepackage{colortbl}
\usepackage{fontawesome5}
\usepackage{adjustbox}
\usepackage{xspace}
\usepackage{placeins}
\usepackage{subfigure}
\usepackage{tabularx}

\usepackage[disable]{todonotes}

\usepackage{tikz}
\usetikzlibrary{automata}
\usetikzlibrary{arrows.meta}
\usetikzlibrary{bending}
\usetikzlibrary{calc}
\usetikzlibrary{positioning}
\usetikzlibrary{shapes.geometric}
\usetikzlibrary{shapes.symbols}
\tikzset{
      >={Stealth[round,bend]},
      StealthDouble/.tip={Stealth[round,quick,length=5.75pt]},
      double>/.style={-StealthDouble,double},
      root/.style={draw,rectangle,rounded corners=3pt},
      inode/.style={draw,circle,inner sep=0pt,minimum height=1.2em},
      low/.style={->,densely dotted},
      high/.style={->},
      termn/.style={draw,rectangle},
      terma/.style={draw,rectangle,double},
      rlink/.style={->,shorten >=3pt},
      win/.style={fill=green!20},
      wine/.style={green, line width=4pt, opacity=0.3},
      automaton/.style={
        shorten >=1pt,>={Stealth[round,bend]},
        node distance=2cm,
        initial text=,
        every initial by arrow/.style={every node/.style={inner sep=0pt}},
      every state/.style={
        align=center,
        fill=white,
        minimum size=7.5mm,
        inner sep=0pt,
        execute at begin node=\strut,
      },%
   }
}

\usepackage{mathrsfs}           
\SetProcNameSty{texttt}
\SetKwProg{Fn}{Function}{}{end}
\SetVlineSkip{0.0em}
\SetKwComment{tcp}{// }{}
\SetKwInOut{Input}{input}
\SetKwInOut{Output}{output}
\SetKwInput{Var}{var}
\SetKwFunction{mk}{makebdd}
\SetKwFunction{apply}{apply2}
\SetKwFunction{applyone}{apply1}
\SetKwFunction{applysc}{apply2sc}
\SetKwFunction{eval}{eval}
\SetKwFunction{fun}{fun}
\SetKwFunction{leaffun}{leaffun}
\SetKwFunction{funsh}{f}
\SetKwFunction{qtb}{quant\_to\_bool}
\SetKwFunction{leaves}{leaves}
\SetKwFunction{solveone}{solve1}
\SetKwFunction{solvetwo}{solve2}
\SetKwFunction{newvertex}{new\_vertex}
\SetKwFunction{newedge}{new\_edge}
\SetKwFunction{freezevertex}{freeze\_vertex}
\SetKwFunction{setwinner}{set\_winner}
\SetKwFunction{realizability}{realizability}
\SetKwFunction{declarevertex}{declare\_vertex}
\SetKwFunction{refine}{refine}
\SetKwFunction{win}{win\textsubscript{\textsc{o}}}
\SetKwFunction{lose}{win\textsubscript{\textsc{i}}}
\SetKwFunction{dfs}{dfs}
\SetKw{Assert}{assert}
\SetKw{SuchThat}{such that}
\SetKw{Continue}{continue}
\SetKw{Break}{break}

\newcommand{\CA}{\mathcal{A}}

\newcommand{\CI}{\mathcal{I}}
\newcommand{\CO}{\mathcal{O}}
\newcommand{\CU}{\mathcal{U}}
\newcommand{\CS}{\mathcal{S}}
\newcommand{\CT}{\mathcal{T}}

\newcommand{\CP}{\mathcal{P}}
\newcommand{\CQ}{\mathcal{Q}}
\newcommand{\BB}{\mathbb{B}}

\newcommand{\MTBDD}{\mathsf{MTBDD}}

\newcommand{\lang}{\mathscr{L}}

\newcommand{\tr}{\mathsf{tr}}
\newcommand{\subs}{\mathsf{sf}}

\newcommand{\LTLf}{{\texorpdfstring{\ensuremath{\mathsf{LTL_f}}}{LTLf}}\xspace}
\newcommand{\ttrue}{\mathit{tt}}
\newcommand{\ffalse}{\mathit{ff}}
\newcommand{\F}{\mathsf{F}} 
\newcommand{\G}{\mathsf{G}} 
\newcommand{\X}{\mathsf{X}} 
\newcommand{\StrongX}{\mathsf{X^{!}}} 
\newcommand{\R}{\mathbin{\mathsf{R}}} 
\newcommand{\U}{\mathbin{\mathsf{U}}} 
\newcommand{\limplies}{\rightarrow}
\newcommand{\liff}{\leftrightarrow}
\newcommand{\lxor}{\oplus}

\newcommand{\QLTLf}{{\texorpdfstring{\ensuremath{\mathsf{QLTL_f}}}{QLTL_f}}\xspace}

\newcommand{\fp}[1]{{\sf fp}(#1)}
\newcommand{\fpObs}[1]{{\sf fpObs}(#1)}
\newcommand{\reach}[1]{{\mathsf{Reach}(#1)}}
\newcommand{\reachB}[1]{{\mathsf{Reach}^{B}(#1)}}

\newcommand{\spotbsc}{\texttt{ltlfsynt-po-naive}\xspace} 
\newcommand{\spototf}{\texttt{ltlfsynt-po}\xspace} 
\newcommand{\syftbsc}{\emph{Syft-bsc}\xspace} 
\newcommand{\syftmso}{\emph{Syft-mso}\xspace} 
\newcommand{\syftp}{\emph{Syft-proj}\xspace} 

\newcommand{\lucasbench}{\emph{TV-Benchmarks}\xspace} 
\newcommand{\syntcompbench}{\emph{SYNTCOMP-fin Benchmarks}\xspace} 

\newcommand{\propequiv}{\sim}

\newcommand{\citet}[1]{\citeauthor{#1}~\shortcite{#1}}

\definecolor{lime}{HTML}{A6CE39}
\DeclareRobustCommand{\orcidicon}{
	\hspace{-2.5mm}
	\begin{tikzpicture}[baseline={(0,-0.12)}]
	\draw[lime, fill=lime] (0,0)
	circle [radius=0.16]
	node[white] (ID) {{\fontfamily{qag}\selectfont \tiny ID}};
	\draw[white, fill=white] (-0.0625,0.095)
	circle [radius=0.007];
	\end{tikzpicture}
	\hspace{-2.5mm}
      }
\def\orcidID#1{\href{https://orcid.org/#1}{\smash{\orcidicon}}}

\title{On-the-fly \LTLf{} Synthesis under Partial Observability}

\author{
Nadav~Alon$^1$  \and
Supratik~Chakraborty$^2$  \and
Alexandre~\mbox{Duret-Lutz}$^3$ \and
Dror~Fried$^1$, \\
Lucas~M.~Tabajara$^4$ \and
Moshe~Y.~Vardi$^5$ \and
Shufang~Zhu$^6$
\affiliations
$^1$The Open University of Israel, Israel\\
$^2$IIT Bombay, India\\
$^3$LRE, EPITA, Le Kremlin-Bicêtre, France\\
$^4$Runtime Verification Inc., USA \\
$^5$Rice University, USA \\
$^6$University of Liverpool, UK \\
}

\begin{document}
\allowdisplaybreaks

\maketitle

\begin{abstract}
    \LTLf synthesis under partial observability requires reasoning about unobservable environment variables, which is typically handled by constructing a belief-state DFA via subset construction that universally quantifies these variables. Existing approaches perform this construction as a separate step prior to game solving, often generating belief states that are unnecessary in practice.
    We propose an on-the-fly approach to \LTLf synthesis under partial observability based on observable progression. Our method incrementally builds the belief-state DFA by progressing the specification with respect to observable variables only, universally quantifying unobservable variables on the fly. We prove the correctness of the construction and show that it naturally enables on-the-fly game solving, leading to a fully on-the-fly synthesis framework. Our implementation leverages DFAs represented using Multi-Terminal Binary Decision Diagrams: a compact representation that has proven highly effective for \LTLf synthesis under full observability. Experimental results demonstrate that our approach significantly outperforms existing methods and further highlight the practical benefits of integrating on-the-fly game solving with belief-state construction.

\end{abstract}

\section{Introduction}

A key challenge in Artificial Intelligence (AI) is enabling intelligent agents to autonomously deliberate their courses of action in order to achieve desired tasks~\cite{Reiter01,GhallabNauTraverso2016}. Reasoning about actions, planning, and sequential decision making are closely connected to Formal Methods for strategic reasoning, such as reactive synthesis~\cite{pnueli.89.popl,finkbeiner.16.dsse,ehlers.17.deds}. In reactive synthesis, an agent operates in an adversarial environment and must construct a strategy that guarantees task achievement under all possible environment behaviors. Within this setting, Linear Temporal Logic on finite traces (\LTLf) synthesis has become a natural and expressive framework for specifying and automatically synthesizing correct-by-construction strategies over finite executions~\cite{degiacomo.15.ijcai}. Moreover, \LTLf synthesis is closely related to (strong) planning for temporally extended goals in fully observable nondeterministic domains~\cite{cimatti.03.aij,bacchus.00.aij,calvanese.02.kr,baier.07.icaps,gerevini.09.aij,degiacomo.18.ijcai,camacho.19.icaps}.


In many realistic settings, however, agents do not have complete information about the environment in which they operate. They must therefore make decisions based on partial observations and reason about uncertainty while interacting with the environment~\cite{rintanen2004complexity}. As such, partial observability has been widely studied in planning and synthesis, including contingent and belief-based planning as well as reactive synthesis from logic specifications~\cite{BonetGeffner2000,HoffmannBrafman2005,Maliah2022,KupfermanV97,GiacomoV16,TV2020}.



In this paper, we focus on \LTLf synthesis under partial observability as described by~\citet{GiacomoV16}. In their setting, a specification is given, in the form of an \LTLf formula over observable environment variables, unobservable environment variables, and system variables, which are also observable. The challenge is to construct a controller that for every history of observable environment assignments, directs how to assign the system variables, inducing a finite trace such that, no matter what the assignments for the unobservable variables are, the overall sequence meets the given specification.

\citet{GiacomoV16} provide a general solution to \LTLf synthesis under partial observability. Given an \LTLf specification, their (belief-state-based) approach first constructs a deterministic finite automaton~(DFA) for the formula. To handle partial observability, the DFA is lifted to a belief-state automaton via a subset construction that universally quantifies over unobservable environment variables. The synthesis problem then reduces to solving a reachability game on the resulting belief-state DFA. An MSO-based approach was suggested by~\citet{TV2020}, which translates the \LTLf specification into a monadic second-order logic (MSO) formula with the unobservable variables universally quantified explicitly -- a second-order operation that exceeds the expressive power of FOL. The approach then constructs a belief-state DFA from the MSO formula. In both approaches, synthesis ultimately requires solving a reachability game on a fully built belief-state DFA.

Through a symbolic implementation and experimental evaluation, \citet{TV2020}
show that the MSO-based approach overall outperforms the other approaches in practice. A key reason for this is a semi-symbolic representation of DFAs~\cite{henrikson.95.tacas}, where transitions are compactly represented using Multi-Terminal Binary Decision Diagrams~(MTBDDs). This representation, called MTDFA in the sequel, allows extensive sharing of common transition structures across states, resulting in significantly more effective automaton construction.

The effectiveness of MTDFA-based representations is further confirmed by recent advances in \LTLf synthesis~(under full observability). Many state-of-the-art \LTLf synthesizers rely on such representations~\cite{zhu.17.ijcai,bansal.20.aaai,degiacomo.21.icaps,zhu.25.tacas,duret.25.ciaa}.
Using MTDFA alone, however, has its limitations. The top-ranked \LTLf synthesizer in the 2025 reactive synthesis competition\footnote{See~\url{https://www.syntcomp.org/syntcomp-2025-results/}.}, \texttt{ltlfsynt}, combines the MTDFA-based representation with \emph{on-the-fly} techniques, incrementally constructing the MTDFA and integrating this construction with game solving~\cite{duret.25.ciaa}. This demonstrates that a tight integration of MTDFA-based on-the-fly automaton construction and on-the-fly game solving can substantially improve performance in practice.

%
%

This naturally raises the question of whether the same combination of techniques can be brought to \LTLf synthesis under partial observability. While MTDFAs provide a powerful technique for compact automaton representation and on-the-fly construction, partial observability introduces an additional challenge: universal quantification over unobservable environment variables. In all existing approaches to \LTLf synthesis under partial observability,  such as that of~\citet{TV2020}, this quantification is handled conceptually as a separate step, rather than being operationally integrated with automaton construction and game solving.

In this paper, we address this challenge by bringing together the key techniques that have proven successful in \LTLf synthesis under full observability, with complete operational integration of MTDFAs  and universal quantification. We first propose an on-the-fly construction of belief-state DFAs from \LTLf specifications, which we call \emph{observable progression}. This technique incrementally builds an automaton for the \LTLf specification via formula progression, while at the same time universally quantifying over unobservable environment variables, therefore directly resulting in a belief-state DFA. This construction enables integrating game solving directly during belief-state DFA construction. Thus, belief states are generated incrementally and only when needed, providing an on-the-fly approach to \LTLf synthesis under partial observability, which still preserves the 2EXPTIME worst-case complexity, as for full-observability. Furthermore, to harness the efficiency of MTDFA, we implement our tool using the MTDFA representation, in which belief states are obtained naturally as MTBDD terminals.
In particular, with careful variable ordering, our MTDFA-based belief-state construction does not incur any additional blowup.
This construction allows us to further integrate on-the-fly game solving on the same MTDFA representation. To summarize, our approach enables the automaton construction, the universal quantification, and the game solving, all to be performed on the fly and directly on MTBDDs.

Our implementation resulted in a modified version of \texttt{ltlfsynt} that we call \spototf, which can now handle \LTLf synthesis under partial observability. Empirical evaluation showed that \spototf dramatically outperforms the tools from~\citet{TV2020}, which so far has represented the state of the art for partial-observability \LTLf synthesis. In addition, we show how each of the phases, belief-state DFA via observable progression, and on-the-fly game solving, contributes to the performance of \spototf. Detailed analysis demonstrates the practical benefits of enabling on-the-fly game solving.
\ifappendix
\fi

\section{Preliminaries}


\subsection{Words over Assignments}

Let $\Sigma$ be a finite alphabet. We view a word over $\Sigma$ as a sequence of letters indexed by positions: an \emph{infinite word} is a function $\sigma:\mathbb{N}\to\Sigma$, whereas a \emph{finite word} of length $n$ is a function $\sigma:\{0,1,\dots,n-1\}\to\Sigma$. The collection of all infinite words is denoted $\Sigma^\omega$. For each $n\in\mathbb{N}$, the set of length-$n$ words is $\Sigma^n$, and we let $\Sigma^\star$ (resp. $\Sigma^+$) range over all finite words of length at least $0$ (resp. strictly positive). For a finite word $\sigma$, its length is written $|\sigma|$. If $\sigma\in\Sigma^n$ and $0\le i<n$, we write $\sigma(..i)$ for the prefix of $\sigma$ ending at position $i$. Note that $\sigma(..i)$ is of length $i+1$. Similarly, we denote $\sigma(i..)$ for the suffix starting at position $i$;
the same notation applies to $\sigma\in\Sigma^\omega$ in the natural way.

Let $\CP$ be a finite set of Boolean variables (equivalently, \emph{atomic propositions}). An \emph{assignment} for $\CP$ is a mapping  of variables to values in
$\BB=\{\bot,\top\}$, represented as $w:\CP\to\BB$.
We use
$\BB^\CP$ to denote the set of all such assignments. When $\CP_1$ and $\CP_2$ are disjoint variable sets, any pair of assignments $w_1\in\BB^{\CP_1}$ and $w_2\in\BB^{\CP_2}$ can be combined into a single assignment over $\CP_1\cup\CP_2$, written $w_1\sqcup w_2$, by agreeing with $w_1$ on $\CP_1$ and with $w_2$ on $\CP_2$. Formally, we have $(w_1\sqcup w_2)(v)=w_1(v)$ for $v\in\CP_1$ and $(w_1\sqcup w_2)(v)=w_2(v)$ for $v\in\CP_2$.

We use these objects to encode synchronous, discrete-time Boolean signal models. Assigning one variable in $\CP$ to each signal, the global behavior over time is represented by a finite word $\sigma\in(\BB^\CP)^+$ whose letters are assignments of $\CP$ at successive time steps. The operator $\sqcup$ is lifted pointwise to such words: for disjoint $\CP_1,\CP_2$ and words $\sigma_1\in(\BB^{\CP_1})^n$, $\sigma_2\in(\BB^{\CP_2})^n$ of the same length $n$, we define $\sigma_1\sqcup\sigma_2\in(\BB^{\CP_1\cup\CP_2})^n$ by $(\sigma_1\sqcup\sigma_2)(i)=\sigma_1(i)\sqcup\sigma_2(i)$ for every $i$ such that $0\le i<n$.

\subsection{Linear Temporal Logic over Nonempty Finite Words.}

We use classical $\LTLf$ semantics over nonempty finite
words~\cite{degiacomo.13.ijcai}.

\begin{definition}[$\LTLf$ formulas]
An $\LTLf$ formula $\varphi$ is built from a set $\CP$ of atomic propositions,
using the following grammar, where $p\in\CP$, and $\odot\in\{\land,\lor,\limplies,\liff,\ldots\}$ is any Boolean operator:
  $
  \varphi ::= \ttrue\mid\ffalse \mid p \mid \lnot\varphi \mid \varphi\odot\varphi \mid
  \X\varphi \mid \StrongX\varphi \mid \F\varphi\mid \G \varphi\mid \varphi\U\varphi \mid \varphi\R\varphi
  $.\looseness=-1

Symbols $\ttrue$ and $\ffalse$ represent the \emph{true} and \emph{false} $\LTLf$ formulas. Other than the usual Boolean operators, we have the temporal operators $\X$ (weak next), $\StrongX$ (strong next), $\F$~(finally), $\G$ (globally), $\U$ (until), and $\R$ (release). Let $\LTLf(\CP)$ denote the set of all formulas produced by the above grammar. For $\varphi \in \LTLf(\CP)$, we use $\subs(\varphi)$ to denote the set of all sub-formulas of $\varphi$. A \emph{maximal temporal subformula} of $\varphi$ is a subformula whose primary operator is temporal and that is not strictly contained within any other temporal subformula of $\varphi$.

The semantics of \LTLf{} defines when a formula $\varphi\in\LTLf(\CP)$ is satisfied by a word $\sigma\in (\BB^\CP)^n$ at position $0\le i< n$.  We denote this by $\sigma,i\models\varphi$, and define it as follows.
\begin{gather*}
\begin{aligned}
  \sigma,i \models \ttrue &\iff i<n & \sigma,i\models p &\iff p\in \sigma(i) \\
  \sigma,i\models \ffalse &\iff i=n\quad\quad & \sigma,i\models \lnot\varphi &\iff \lnot(\sigma,i\models \varphi)
\end{aligned}\\
\begin{aligned}
  \sigma,i\models \varphi_1\odot\varphi_2 &\iff (\sigma,i\models \varphi_1)\odot(\sigma,i\models \varphi_2)\\
   \sigma,i\models \X\varphi &\iff (i+1=n)\lor(\sigma,i+1\models \varphi)\\
    \sigma,i\models \StrongX\varphi &\iff (i+1<n)\land(\sigma,i+1\models \varphi)\\
    \sigma,i\models \F\varphi &\iff \exists j\in [i,n),\,\sigma,j\models \varphi\\
    \sigma,i\models \G\varphi &\iff \forall j\in [i,n),\,\sigma,j\models \varphi\\
    \sigma,i\models \varphi_1\U\varphi_2 &\iff \\ &\mathllap{\exists j{\in}}[i,n),\big((\sigma,j\models \varphi_2)\land
     (\forall k{\in} [i,j),\,\sigma,k\models \varphi_1)\big)\\
    \sigma,i\models \varphi_1\R\varphi_2 &\iff \\
    & \mathllap{\forall j{\in}}[i,n),\big((\sigma,j\models \varphi_2)\lor(\exists k{\in} [i,j),\,\sigma,k\models \varphi_1)\big)
 \end{aligned}
\end{gather*}

The set of words that satisfy $\varphi\in\LTLf(\CP)$ is represented by $\lang(\varphi)=\{\sigma\in (\BB^\CP)^+ \mid \sigma,0\models \varphi\}$.
\end{definition}

\begin{example}\label{ex:psi}
  Consider the following $\LTLf$ formula over
  $\CP=\{u,i,o\}$:
  $\Psi = \big(\G\F u \rightarrow \F(i \leftrightarrow o)\big) \land \big(\G\F\neg u \rightarrow \F(i \lor o)\big)$.
    Suppose $u$ and $i$ are inputs, and $o$ is the output of a system.
    The above formula specifies that if $u$ is set in the last step, then $o$ must eventually have the same value as $i$. Otherwise, at least one of $i$ or $o$ must eventually be set, although it is not necessary for them to have the same value.
  \end{example}

\begin{definition}[Propositional Equivalence~\cite{esparza.18.lics}]\label{def:propequiv}
    For $\varphi\in\LTLf(\CP)$, let $\varphi_P$ be the
    Boolean formula obtained from $\varphi$ by replacing every
    maximal temporal subformula $\psi$ by a Boolean variable $x_\psi$.
    Two formulas $\varphi_1,\varphi_2\in\LTLf(\CP)$ are 
    \emph{propositionally equivalent}, denoted $\varphi_1\propequiv \varphi_2$, if
    ${\varphi_1}_P$ and ${\varphi_2}_P$ are semantically equivalent Boolean formulas.
  \end{definition}
  Note that if ${\varphi_1} \propequiv {\varphi_2}$, then $\lang(\varphi_1) = \lang(\varphi_2)$, but the converse is not true in general. We use $[\varphi]_\propequiv\in\LTLf(\CP)$ to denote some unique representative of the equivalence class of $\varphi$ with respect to $\propequiv$.


\subsection{\LTLf{} Synthesis and Realizability}


Consider a setting where the desired behavior of a system interacting with an environment is specified by an \LTLf{} formula over variables representing inputs of the environment and outputs of the system. The synthesis problem asks us to design a \emph{controller}~(aka. \emph{strategy}/\emph{plan}) that observes the history of inputs and determines the outputs at every time step, such that the specification is satisfied for every input sequence produced by the environment.
Not all system inputs may be observable or even usable by the controller.  This can happen if some inputs are noisy or cannot be measured reliably enough.  Therefore, it makes sense to partition the input variables into two disjoint sets $\CI$ and $\CU$, representing observable and unobservable
inputs, respectively.  Let $\CO$ be the set of output variables.  In general, the desired
system behavior is specified by a formula $\varphi \in \LTLf(\CI \uplus \CU \uplus \CO)$.


We use the \emph{terminating transducer} interpretation of controllers, originally introduced by~\citet{bansal.23.atva}.  As observed by~\citet{jacobs.23.tlsf12}, this is a natural interpretation in the context of \LTLf{} synthesis, since we expect a controller to know when a task is completed, and therefore when to stop interacting with the environment. The controller itself may be represented as a deterministic Mealy or Moore machine
with input alphabet $\BB^\CI$, output alphabet $\BB^\CO$, and a distinguished subset of states designated as \emph{terminating states}. With Mealy semantics, the controller can access the past and also current inputs to determine its current output; with Moore semantics, it can access only the past inputs. The controller $M$ induces a function $\rho_M: (\BB^\CI)^*\to(\BB^\CO)$ that maps a history of assignments of observable input variables to the current assignment of output variables. Given an input word $\sigma\in(\BB^\CI)^n$ of length $n$, the controller produces a corresponding word of $n$ output assignments, denoted $\sigma_{\rho_M}\in(\BB^\CO)^n$. If $\sigma$ leads the underlying machine to a terminating state, the system is assumed to stop interacting with the environment after reading $\sigma$.

\noindent{\bfseries Synthesis under full observability:}
To understand when a controller $M$ \emph{realizes} a specification $\varphi$, we first consider the case when all inputs of the environment are observable by the controller, i.e., $\CU = \emptyset$.

\begin{definition}[\cite{bansal.23.atva,jacobs.23.tlsf12}]\label{def:fo-realizability}
A controller $M$ realizes an \LTLf{} specification $\varphi$ if for every word $\sigma\in(\BB^\CI)^\omega$ there exists a position $k \ge 0$ such that (a) $M$ enters a terminating state after reading the input sequence $\sigma(..k)$, and (b) $(\sigma\sqcup\sigma_{\rho_M})(..k)\in\lang(\varphi)$.
\end{definition}

If we view the interaction between the system and the environment as a two-player game, then Definition~\ref{def:fo-realizability} implies the following. For every input sequence produced by the environment, the controller generates an output at each step based on the input history so far. It does so such that, after finitely many steps, the play reaches a terminating state of the transducer.  When the resulting finite interaction sequence satisfies the \LTLf{} specification $\varphi$, the system player chooses to terminate the play.  The first player is chosen according to the desired semantics (Mealy or Moore).

\noindent {\bfseries Synthesis under partial observability:}
If $\CU \neq \emptyset$, the situation becomes more interesting. While the controller $M$ generates its output sequence by observing only the inputs on $\CI$, satisfaction of $\varphi$ depends on the values on $\CI$, $\CO$ and $\CU$. This motivates the following definition of realizability under partial observation~\cite{GiacomoV16}.


\begin{definition}\label{def:po-realizability}
A controller $M$ implementing the function $\rho_M: (\BB^\CI)^* \rightarrow \BB^\CO$ realizes a specification $\varphi \in \LTLf(\CI \uplus \CU \uplus \CO)$ if for any word $\sigma\in(\BB^\CI)^\omega$ there exists a position $k \ge 0$ such that (a) $M$ enters a terminating state after reading the input sequence $\sigma(..k)$, and (b) for every $\sigma_{\CU} \in (\BB^\CU)^k$, the word $\sigma_{\CU} \sqcup (\sigma \sqcup \sigma_{\rho_M})(..k) \in \lang(\varphi)$.
\end{definition}

The alternation of quantifiers in the above definition deserves careful attention. Intuitively, once the controller decides to terminate the play (by entering a terminating state), it must be possible to satisfy $\varphi$ by augmenting $(\sigma \sqcup \sigma_{\rho_M})(..k)$ with any arbitrary sequence of valuations of $\CU$ of length $k+1$. Note that we do not pick a separate $k$ for each word in $(\BB^\CU)^\omega$. Instead, once $k$ is chosen, the same choice works in condition (b) of Definition~\ref{def:po-realizability} for all words in $(\BB^\CU)^k$. This is why, despite the lack of visibility of $\CU$, the outputs generated by the controller allow the system  to win the game, irrespective of how the environment assigns values to $\CU$.

We say a specification $\varphi \in \LTLf(\CI \uplus \CU \uplus \CO)$ is \emph{realizable under full (resp. partial) observability} if there exists a controller $M$ that realizes $\varphi$ in the sense of Definition~\ref{def:fo-realizability}~(resp. Definition~\ref{def:po-realizability}). The synthesis problem for \LTLf{} under full~(resp. partial) observability is the task of algorithmically generating a controller $M$ that realizes $\varphi$.
The computational complexity of both problems is known to be 2EXPTIME-complete \cite{degiacomo.15.ijcai,GiacomoV16}.

Referring to Example~\ref{ex:psi}, if both inputs $i$ and $u$ are observable, the specification is realizable in one step by a Mealy machine controller; however, with only $i$ observable, two steps are needed.  We  discuss this in detail in Section~\ref{sec:impl}.
%

\noindent {\bfseries Some important notions used in later sections:} The notions of formula progression~\cite{xiao.21.aaai,degiacomo.22.ijcai,xiao.25.tosem,duret.25.ciaa}, belief-state automata construction~\cite{GiacomoV16}, and on-the-fly interleaved construction of automata and game-solving~\cite{xiao.21.aaai,degiacomo.22.ijcai,xiao.2024.vmcai,favorito.23.rcra,li.25.ecai,duret.25.ciaa} have been used in earlier work in the context of \LTLf synthesis under full observability.  We introduce non-trivial adaptations of these techniques to work in the presence of partial observability.
Details of these techniques and our adaptations are discussed in Section~\ref{sec:otf}.

Multi-Terminal Binary Decision Diagrams (MTBDDs)~\cite{minato.96.vlsi,fujita.97.fmsd,klarlund.01.tr},
also called Algebraic Decision Diagrams (ADDs)~\cite{bahar.93.iccad,somenzi.15.cudd}, are graph-based symbolic representations of functions $f:\BB^\CP\to\CS$, where $\CP$ is a set of Boolean propositions and $\CS$ is a set of labels. We use MTBDDs to symbolically represent the next-state transition functions of deterministic finite automata (DFA) arising from \LTLf{} formulas.  MTBDD-based representations of DFA, also called MTDFA, have been used by~\citet{zhu.17.ijcai} and~\citet{duret.25.ciaa} for \LTLf{} synthesis under full observability. More details on MTDFAs are given in Section~\ref{sec:impl}.\looseness=-1

\section{On-the-fly Synthesis Framework}\label{sec:otf}

This section presents an on-the-fly approach to \LTLf synthesis under partial observability. We first briefly recall formula progression as a standard technique for constructing DFAs for \LTLf specifications on the fly. We then extend this technique to partial observability by introducing \emph{observable progression}, which incrementally constructs a belief-state DFA, performing universal quantification over unobservable variables at each step. Finally, we show how this construction enables on-the-fly game solving, by interleaving belief-state DFA construction with game exploration.

\subsection{\LTLf-to-DFA via Formula Progression}\label{sec:fp}

On-the-fly DFA construction for an \LTLf formula is based on the idea of formula progression~\cite{xiao.21.aaai,degiacomo.22.ijcai,xiao.25.tosem}. Our formula progression procedure, denoted by $\fp{\cdot,\cdot}$, is adapted from existing progression-based translations to better align with state-of-the-art MTDFA-based implementation of formula progression~\cite{duret.25.ciaa}. Specifically, our adaptation requires $\fp{\cdot,\cdot}$ to return a (formula, flag) pair, rather than a formula alone. This is essential for introducing \emph{observable progression} in the next section, which leverages the MTDFA representation to construct belief-state DFAs on the fly.

The procedure $\fp{\varphi, w}$ takes an \LTLf formula $\varphi$ and an assignment $w \in \BB^{\CP}$, and returns a pair $(\varphi', b) \in \LTLf(\CP) \times \BB$, where $\varphi'$ is an \LTLf formula representing the remaining to be satisfied after reading $w$, and $b$ is a Boolean value indicating whether $\varphi$ is satisfied after reading $w$, assuming the trace terminates at that point. Note that we use the following lifting of Boolean connectives to pairs: for $\odot \in \{\land,\lor,\limplies,\liff,\lxor\}$, we have $(\varphi_1,b_1) \odot (\varphi_2,b_2) = ([\varphi_1\odot\varphi_2]_\propequiv, b_1\odot b_2)$, and $\neg(\varphi,b) = ([\neg\varphi]_\propequiv, \neg b)$.
\begin{align*}
&\fp{\ttrue, w} = (\ttrue, \top) \quad \text{and} \quad \fp{\ffalse, w} = (\ffalse, \bot); \\
&\fp{p, w} =
\begin{cases}
(\ttrue, \top) \quad \text{ if } p \in w,\\
(\ffalse, \bot) \quad \text{ if } p \not\in w;
\end{cases}\\
&\fp{\neg \varphi, w} = \neg \fp{\varphi, w}; \\
&\fp{\varphi_1 \odot \varphi_2, w}
= \fp{\varphi_1, w} \odot \fp{\varphi_2, w}; \\
&\fp{\X \varphi, w} = (\varphi, \top) \quad \text{and} \quad
\fp{\StrongX \varphi, w} = (\varphi, \bot);\\
&\fp{\F\varphi, w}
= \fp{\varphi, w} \vee
\bigl(\F\varphi, \bot \bigr); \\
&\fp{\G\varphi, w}
= \fp{\varphi, w} \wedge
\bigl(\G\varphi, \top\bigr); \\
&\fp{\varphi_1 \U \varphi_2, w}
= \fp{\varphi_2, w} \vee
\bigl(\fp{\varphi_1, w} \wedge (\varphi_1 \U \varphi_2, \bot)\bigr); \\
&\fp{\varphi_1 \R \varphi_2, w}
= \fp{\varphi_2, w} \wedge
\bigl(\fp{\varphi_1, w} \vee (\varphi_1 \R \varphi_2, \top)\bigr).
\end{align*}

In order to build a DFA of an \LTLf formula $\varphi \in \LTLf(\CP)$ via formula progression, we consider $[\varphi]_\propequiv$ as the initial state of the DFA.
Then, for each assignment $w \in \BB^{\CP}$, we compute $\fp{\varphi, w}$ and obtain a progressed formula, which is treated as the successor state of $\varphi$ with respect to the input $w$. By repeating this construction for newly generated formulas, the DFA is built on the fly, generating only the states that are reachable from $\varphi$, denoted by $\reach{\varphi}$, where $\reach{\varphi} = \{[\varphi]_\propequiv\} \cup \{[\varphi']_\propequiv \mid (\varphi', b) = \fp{\varphi, \sigma} \text{ for every } \sigma \in (\BB^\CP)^{+} \}$. Note that the function $\fp{\cdot,\cdot}$ is naturally generalized to finite non-empty words $\sigma \in (\BB^{\CP})^{+}$: $\fp{\varphi,\sigma}=
(\varphi_{|\sigma|-1},b_{|\sigma|-1})   $ with
$(\varphi_i,b_i)=\fp{\varphi_{i-1},\sigma(i)}$ for $0 \leq i < |\sigma|$, and $\varphi_{-1}$ refers to the original formula $\varphi$.
Due to propositional equivalence, this construction produces a doubly exponential DFA, i.e. $2^{2^{O(|\subs(\varphi)|)}}$ states, in the worst case. 

\subsection{Belief-State DFA Construction}


Note that in \LTLf synthesis under partial observability, a subset of variables is unobservable. Therefore, a straightforward way to handle the unobservable variables $\CU$ is projecting them away from the transitions of the DFA over alphabet $\CI \uplus \CO \uplus \CU$. This projection introduces an NFA over alphabet $\CI \uplus \CO$ that must be determinized via subset construction to obtain a belief-state DFA~\cite{GiacomoV16}. In this section, we show how this belief-state DFA construction can be performed on the fly by interleaving formula progression with subset construction.

In the remainder of the section, for simplicity, we only distinguish observable and unobservable variables. Thus, we consider $\CP = \CP_o \uplus \CP_u$, where $\CP_o = \CI \uplus \CO$ is the set of observable (environment and system) variables and $\CP_u = \CU$ is the set of unobservable (environment) variables.

Intuitively, formula progression decomposes an \LTLf formula $\varphi$ into a requirement about the \emph{current} assignment, which can be checked immediately, and a requirement about the \emph{future} that must hold on the yet unavailable suffix of the trace. Under partial observability, however, the current assignment is only partially known, due to the unobservable variables $\CP_u$. As a result, progressing an \LTLf formula with respect to an observable assignment $w_o \in \BB^{\CP_o}$ no longer leads to a single future requirement, but to a set of possible future requirements corresponding to different assignments of the unobservable variables.

Accordingly, given an \LTLf formula $\varphi \in \LTLf(\CP)$, a belief state is a finite set $s = \{\varphi_1, \ldots, \varphi_k\} \subseteq \reach{\varphi}$. In particular, we represent such a belief-state by the formula $\psi_s = \bigwedge_{\varphi_i\in s}\varphi_i$, which naturally captures the universal quantification over unobservable variables induced by subset construction.
It is worth noting that the implementation of belief states by \citet{TV2020} did introduce an additional exponential blowup beyond standard \LTLf-to-DFA construction, resulting in a triply exponential worst-case complexity. However, this additional cost is not inherent to our belief-state approach.
Although belief states are defined as finite sets of formulas, they can be represented as conjunctions of formulas, and identified through propositional equivalence, rather than treated as distinct syntactic sets. We return to this point later and formally show that the resulting belief-state space remains doubly exponential in the size of the original formula.
%
%
For simplicity, we identify belief states $s$ with their conjunctive formula representations $\psi_s$, or $\psi$ when it's clear from the context.


\medskip
\noindent\textbf{Observable Progression.} Let $\psi_s$ be a belief state, and let $w_o \in \BB^{\CP_o}$ be an observable assignment. Intuitively, observable progression takes the current belief state, progresses it with respect to the observable assignment $w_o$, and produces the corresponding successor belief state. To this end, observable progression considers all full assignments $w \in \BB^{\CP}$ that agree with $w_o$ on the observable variables, i.e., $w_o = w|_{\CP_o}$. Formula progression is applied to $\psi_s$ under each such assignment. Since the successor belief state must account for all such possible assignments, we combine the resulting progressions conjunctively. Note that formula progression returns a pair in $\LTLf(\CP)\times\BB$, so to define observable progression, we lift conjunction to such pairs as follows: for a non-empty finite set $X \subseteq \LTLf(\CP) \times \BB$, we define $\bigwedge X = \left(
  \bigwedge_{(\varphi,b)\in X} \varphi,
  \ \bigwedge_{(\varphi,b)\in X} b
\right)$.

\begin{definition}\label{def:obsprog}
The observable progression of $\psi_s$ under $w_o$ is defined as
\[
\fpObs{\psi_s, w_o}
=
\bigwedge \bigl\{ \fp{\psi_s, w} \mid w \in \BB^{\CP},\; w|_{\CP_o} = w_o \bigr\}.
\]
\end{definition}

The following two lemmas show that observable progression is a suitable generalization of formula progression to the setting of partial observability. In particular, observable progression preserves the semantics of \LTLf under unobservable variables as well as the propositional behavior of \LTLf formulas.


\begin{lemma}\label{lem:obs-fp}
Let $\CP=\CP_o\uplus\CP_u$, let $\psi_s$ be a belief state, and $\psi_s'$ the observably progressed formula over $w_o\in\BB^{\CP_o}$, i.e., $(\psi_s',\_) = \fpObs{\psi_s,w_o}$.
Let $\sigma\in(\BB^{\CP})^{+}$ be a full trace and let $i<|\sigma|-1$ be a position such that $\sigma(i)|_{\CP_o}=w_o$.
Then
\[
\sigma,i+1 \models \psi_s'
\quad\text{iff}\quad
\forall w_u\in\BB^{\CP_u}:\ \sigma^{w_u},i \models \psi_s,
\]
where $\sigma^{w_u}$ is obtained from $\sigma$ by changing only position $i$:
$\sigma^{w_u}(i)=w_o\sqcup w_u$ and $\sigma^{w_u}(j)=\sigma(j)$ for all $j\neq i$.
\end{lemma}

\begin{proof}
    By definition of observable progression,
    \begin{align*}
    \fpObs{\psi_s,w_o}
    & =
    \bigwedge\bigl\{\,\fp{\psi_s,w}\ \bigm|\ w\in\BB^{\CP},\ w|_{\CP_o}=w_o\,\bigr\} \\
    & =
    \bigwedge\bigl\{\,\fp{\psi_s,\,w_o\sqcup w_u}\ \bigm|\ w_u\in\BB^{\CP_u}\,\bigr\}.
    \end{align*}

	For each $w_u\in\BB^{\CP_u}$, let $(\psi_{w_u},b_{w_u})$ be the result of progressing $\psi_s$ under the assignment $w_o\sqcup w_u$, that is, $(\psi_{w_u},b_{w_u})=\fp{\psi_s,\,w_o\sqcup w_u}$. By construction, the formula $\psi'$ returned by $\fpObs{\psi_s,w_o}$ is the conjunction of all formulas $\psi_{w_u}$ over $w_u\in\BB^{\CP_u}$. Therefore, we need to show
    \begin{equation}
        \sigma,i+1 \models \psi'_s
        \quad\text{iff}\quad
        \forall w_u\in\BB^{\CP_u}:\ \sigma,i+1 \models \psi_{w_u}.
        \tag{1}
    \end{equation}

	Now let $w_u\in\BB^{\CP_u}$ be arbitrary. By definition of the trace $\sigma^{w_u}$, its value at position $i$ is $w_o\sqcup w_u$. By the correctness of formula progression, we have
    \begin{equation}
        \sigma^{w_u},i \models \psi_s
        \quad\text{iff}\quad
        \sigma^{w_u},i+1 \models \psi_{w_u}.
        \tag{2}
    \end{equation}
    Moreover, the traces $\sigma^{w_u}$ and $\sigma$ share the same suffix starting at $i+1$. Hence, the satisfaction of $\psi_{w_u}$ at
    position $i+1$ is the same on both traces. That is,
    \begin{equation}
        \sigma^{w_u},i+1 \models \psi_{w_u}
        \quad\text{iff}\quad
        \sigma,i+1 \models \psi_{w_u}.
        \tag{3}
    \end{equation}

    Combining (2) \& (3), we obtain that for every $w_u\in\BB^{\CP_u}$,
    \[
    \sigma^{w_u},i \models \psi_s
    \quad\text{iff}\quad
    \sigma,i+1 \models \psi_{w_u}.
    \]
    Taking universal quantification over all $w_u\in\BB^{\CP_u}$ and using
    (1), we conclude that
    \[
    \sigma,i+1 \models \psi'_s
    \quad\text{iff}\quad
    \forall w_u\in\BB^{\CP_u}:\ \sigma^{w_u},i \models \psi_s.\qedhere
    \]
\end{proof}

\begin{lemma}\label{lem:obs-propequiv}
	Let $\varphi$ and $\psi$ be two \LTLf formulas over $\CP$ such that $\varphi \propequiv \psi$, and let $\varphi'$ and $\psi'$ be their respective observably progressed formulas over $w_o\in\BB^{\CP_o}$, i.e., $(\varphi',\_) = \fpObs{\varphi,w_o}$ and $(\psi',\_) = \fpObs{\psi,w_o}$. Then
    \[
    \varphi' \propequiv \psi'.
    \]
\end{lemma}

\begin{proof}
By definition of observable progression,
    \[
    \fpObs{\varphi,w_o}
    =
    \bigwedge\bigl\{\,\fp{\varphi,w}\ \bigm|\ w\in\BB^{\CP},\ w|_{\CP_o}=w_o\,\bigr\},
    \]
    and
    \[
    \fpObs{\psi,w_o}
    =
    \bigwedge\bigl\{\,\fp{\psi,w}\ \bigm|\ w\in\BB^{\CP},\ w|_{\CP_o}=w_o\,\bigr\}.
    \]

    Since $\varphi \propequiv \psi$ and formula progression preserves
    propositional equivalence, for every $w\in\BB^{\CP}$ we
    have $\fp{\varphi,w}\propequiv \fp{\psi,w}.$
    For each
    $w\in\BB^{\CP}$, write $(\varphi_w,\_) = \fp{\varphi,w}$ and
    $(\psi_w,\_) = \fp{\psi,w}$, and we have that $\varphi_w \propequiv \psi_w$.
%
    Therefore, since $\varphi_w \propequiv \psi_w$ for every $w\in\BB^{\CP}$ with $w|_{\CP_o}=w_o$, and propositional equivalence is preserved under conjunction, it follows that 
    \[
    \bigwedge_{w|_{\CP_o}=w_o} \varphi_w
    \;\propequiv\;
    \bigwedge_{w|_{\CP_o}=w_o} \psi_w.
    \]
    Hence, $\varphi' \propequiv \psi'$, concluding the proof.
\end{proof}

Let $\varphi \in \LTLf(\CP)$ be an \LTLf formula such that $\CP = \CP_o \uplus \CP_u$. We construct a belief-state DFA for $\varphi$ by using observable progression to obtain the state space and transition function. Intuitively, each state of the automaton represents a belief state reachable from the initial specification $\varphi$ with some sequence of observable assignments $\sigma_o \in (\BB^{\CP_o})^+$. Progressing on an observable assignment $w_o \in \BB^{\CP_o}$ transits from the current belief state to its successor by applying observable progression, which accounts for all full assignments $w \in \BB^\CP$ consistent with $w_o$. Acceptance is checked by the Boolean value returned by observable progression. We generalize \LTLf observable progression from single instance to finite traces by defining $\fpObs{\varphi, \sigma_o\cdot w_o} = \fpObs{\fpObs{\varphi, \sigma_o}, w_o}$, where $w_o \in \BB^{\CP_o}$ and $\sigma_o \in (\BB^{\CP_o})^+$ and
formalize the belief-state DFA construction as follows.

\begin{definition}[\LTLf to Belief-State DFA]\label{def:ltlf2bdfa}
    Given an \LTLf formula $\varphi \in \LTLf(\CP)$, where $\CP = \CP_o \uplus \CP_u$, let ${\CA^B_{\varphi}} = \langle \CS, \BB^{\CP_o}, \iota^B, \delta^B, \CT^B \rangle$ be a belief-state DFA by setting $\CS = \reachB{\varphi}$, where $\reachB{\varphi} = \{[\varphi]_\propequiv\} \cup \{[\psi']_\propequiv \mid (\psi',\_) = \fpObs{\varphi, \sigma_o}, \sigma_o \in (\BB^{\CP_o})^{+}\}$, $\iota^B = [\varphi]_\propequiv$, $\delta^B: \CS \times \BB^{\CP_o} \to \CS $ is such that $\delta^B(s, w_o)= s'$, where $(s',\_) = \fpObs{s,w_o}$ for $s \in \CS$ and $\CT^B = \{(s, w_o) \in \CS \times \BB^{\CP_o} \mid (\_,b) = \fpObs{s,w_o} \text{ and } b = \top \}$.
\end{definition}


The following theorem shows the correctness of the construction in Definition~\ref{def:ltlf2bdfa}.
\begin{theorem}\label{thm:bdfaCorrectness}
	Let $\varphi \in \LTLf(\CP)$ with $\CP=\CP_o\uplus\CP_u$, and let
	$\CA_\varphi^B$ be the belief-state DFA constructed in
	Definition~\ref{def:ltlf2bdfa}. Then for every observable word
	$\sigma_o \in (\BB^{\CP_o})^{+}$,
	\begin{align*}
	    \sigma_o \in & \lang(\CA_\varphi^B)
	\quad \text{iff} \\
	& \forall \sigma \in (\BB^{\CP})^{+} \text{ such that }\sigma|_{\CP_o}=\sigma_o \text{ we have that } \sigma \models \varphi.
	\end{align*}
\end{theorem}

\begin{proof}
	Let $\sigma_o = w_o^0\cdots w_o^{n-1}\in(\BB^{\CP_o})^{+}$. Consider the~(unique) run of $\CA_\varphi^B$ on $\sigma_o$:
	\[
	s_0 \xrightarrow{w_o^0} s_1 \xrightarrow{w_o^1}\cdots
	\xrightarrow{w_o^{n-1}} s_n,
	\]
	where $s_0=[\varphi]_\propequiv$ and
	$\delta^B(s_t,w_o^t)=s_{t+1}$ for $0 \leq t \leq n-1$. By Definition~\ref{def:ltlf2bdfa}, this
	means that
	\[
	(s_{t+1},\_) = \fpObs{s_t,w_o^t}.
	\tag{1}
	\]




	For $0 \leq t < n$, let $\tau \in (\BB^{\CP})^{n-t}$ be such that $\tau|_{\CP_o}=w_o^t\cdots w_o^{n-1}$. We need to show that the following holds:
    \begin{align*}
        &\tau \models s_t \iff \\
        &\forall (w_u^0\cdots w_u^{t-1}) \in (\BB^{\CP_u})^t. (\sigma_o(..t-1) \sqcup (w_u^0\cdots w_u^{t-1}))\cdot\tau \models \varphi.
    \end{align*}
%
	For $t=0$, the equivalence holds since $s_0=[\varphi]_\propequiv$.

	Assume the equivalence holds for some $t<n$. Let
	$w_o=w_o^t=\tau(0)|_{\CP_o}$. By (1), we have
	$(s_{t+1},\_) = \fpObs{s_t,w_o}$. Applying
	Lemma~\ref{lem:obs-fp}, we have that
	\[
	\tau(1..) \models s_{t+1}
	\quad\text{iff}\quad
	\forall w_u\in\BB^{\CP_u}:\ \tau^{w_u} \models s_t,
	\]
	where $\tau^{w_u}(1..) = \tau(1..)$ and $\tau^{w_u}(0)=w_o\sqcup w_u$.

    By induction hypothesis, we have that $\tau(1..) \models s_{t+1}$ iff $\forall (w_u^0\cdots w_u^{t-1}) \in (\BB^{\CP_u})^t$ and $\forall w_u \in (\BB^{\CP_u})$
    \[
    (\sigma_o(..t-1) \sqcup (w_u^0\cdots w_u^{t-1}))\cdot (w_o\sqcup w_u) \cdot \tau(1..) \models \varphi
    \]
    which is equivalent to $\forall (w_u^0\cdots w_u^{t}) \in (\BB^{\CP_u})^{t+1}$
    \[
    (\sigma_o(..t) \sqcup (w_u^0\cdots w_u^{t})) \cdot \tau(1..) \models \varphi.
    \]
    Thus, the equivalence holds for $t+1$.


	Finally, by Definition~\ref{def:ltlf2bdfa}, the run on $\sigma_o$ is
	accepting iff the last transition $(s_{n-1},w_o^{n-1})$ is accepting,
	i.e., iff for $(\_,b)=\fpObs{s_{n-1},w_o^{n-1}}$ we have $b=\top$.
	Unfolding the definition of $\fpObs{\cdot}$ and the Boolean value of $\fp{\cdot}$, this
	holds exactly when $\varphi$ is satisfied for all full
	traces consistent with $\sigma_o$. Therefore,
	$\sigma_o\in\lang(\CA_\varphi^B)$ iff every full trace
	$\sigma\in(\BB^{\CP})^{+}$ with $\sigma|_{\CP_o}=\sigma_o$ satisfies
	$\varphi$.
\end{proof}

We now consider the size of the belief-state space induced by observable progression. Although belief states are defined as sets of formulas obtained from formula progression, they are always identified by propositional equivalence through their conjunctive representation. Intuitively, while belief states may be considered as sets, they are represented as single Boolean formulas built from subformulas of $\varphi$. Indeed, there are at most doubly exponentially many propositionally non-equivalent Boolean formulas of this form, thereby avoiding a triply exponential worst-case blowup of~\citet{TV2020}'s construction.
Thanks to propositional equivalence, this new belief-state construction maintains a doubly exponential worst-case bound, as formalized by the following theorem.

\begin{theorem}\label{thm:bdfaSize}
    Let $\varphi \in \LTLf(\CP)$ and let $\CA_\varphi^B$ be the belief-state DFA constructed in Definition~\ref{def:ltlf2bdfa}. Then the state set
    $\CS$ of $\CA_\varphi^B$ is finite and satisfies
    $|\CS| = |\reachB{\varphi}| = O(2^{2^{n}}),$
    where $n = |\subs(\varphi)|$.
\end{theorem}
\begin{proof}
    By construction, every belief state is represented by a formula obtained as a conjunction of formulas produced by formula progression starting from $\varphi$. Formula progression introduces only Boolean combinations of subformulas of $\varphi$. Since $\subs(\varphi)$ is finite with $|\subs(\varphi)|=n$, there are at most $2^{2^{n}}$ propositional equivalence classes of such Boolean formulas. Therefore, the set of reachable belief states, considering propositional equivalence~(Lemma~\ref{lem:obs-propequiv}), is finite and bounded by $O(2^{2^{n}})$.
\end{proof}

A straightforward way of using our belief-state DFA construction technique would be solving a reachability game on a fully built belief-state DFA, as done by~\citet{TV2020}. However, this cannot avoid the worst-case doubly exponential blowup. In the next section, we show how to integrate on-the-fly game solving to the belief-state DFA construction, enabling a unified on-the-fly approach to \LTLf synthesis under partial observability.

\subsection{On-the-fly Synthesis}\label{sec:otf-game}

For \LTLf synthesis under full observability, on-the-fly synthesis techniques have proven to be highly effective~\cite{xiao.21.aaai,degiacomo.22.ijcai,xiao.2024.vmcai,favorito.23.rcra,li.25.ecai,duret.25.ciaa}. In this setting, the DFA for the \LTLf specification is constructed incrementally via formula progression and explored only as needed during game solving, avoiding the upfront construction of the full automaton.

Our belief-state construction enables the same on-the-fly technique to be applied to \LTLf synthesis under partial observability. Since belief states and transitions are generated on demand via observable progression, the belief-state DFA does not need to be built upfront. Instead, belief-state construction and game solving can be tightly interleaved, exploiting only those belief states that are actually visited during the reachability game.

\section{Implementation}\label{sec:impl}


The Spot library~\cite{duret.22.cav} contains a tool for \LTLf-synthesis called \texttt{ltlfsynt}\footnote{\url{https://spot.lre.epita.fr/ltlfsynt.html}}~\cite{duret.25.ciaa}.  Our implementation is an extension of \texttt{ltlfsynt} that adds the option \texttt{--unobservable-ins}, which is publicly available since version 2.15 of Spot.
We first recall how that
tool does synthesis with \emph{full observability} before discussing
our changes.

In \texttt{ltlfsynt}, the \LTLf{} specification is converted into a DFA with transition-based acceptance, represented using MTBDDs: they call that representation an MTDFA.
The states $\CQ\subseteq\LTLf(\CP)$ of an MTDFA are identified with $\LTLf$ formulas denoting the language recognized from these states.  The set of outgoing transitions of a state are represented by an MTBDD with internal nodes labeled by $\CP$ and terminals labeled by $\LTLf(\CP)\times\BB$. In a terminal $(\alpha,b)$, the destination state is $\alpha$, while $b$ indicates whether the automaton can stop upon reaching this terminal (in other words, $b$ is the acceptance of the transition leading to $\alpha$).

Concretely, outgoing transitions are computed using the function
$\tr: \LTLf(\CP)\to \MTBDD(\CP,\LTLf(\CP)\times\BB)$ defined
inductively as follows, where we represent a terminal labeled by
$(\alpha,\bot)$ (resp. $(\alpha,\top)$) as
\tikz[baseline=(a.base)]\node[termn](a){$\alpha$};
(resp. \tikz[baseline=(a.base)]\node[terma](a){$\alpha$};).
{\small
\begin{align*}
  \tr(\ffalse)&=\tikz[baseline=(0.base)]\node[termn](0){$\ffalse$};&
  \tr(\StrongX\alpha)&=\tikz[baseline=(a.base)]\node[termn](a){$\alpha$};\\
  \tr(\ttrue)&=\tikz[baseline=(a.base)]\node[terma](a){$\ttrue$};&
  \tr(\X\alpha)&=\tikz[baseline=(0.base)]\node[terma](0)  {$\alpha$};\\
  \tr(p)&=\begin{tikzpicture}[baseline=(p.base)]\node[inode](p) at (0,0) {$p$};
    \node[termn,overlay](0) at (-.4,-.8) {$\ffalse$};
    \node[terma](1) at (.4,-.8) {$\ttrue$};
    \draw[low](p) -- (0);
    \draw[high](p) -- (1);
  \end{tikzpicture}\text{~for~}p{\,\in\,}\CP &
\tr(\lnot \alpha)&=\lnot\tr(\alpha)\\
\tr(\alpha\odot\beta)&=\tr(\alpha)\odot\tr(\beta)\mathrlap{\text{~for any~}\odot\in\{\land,\lor,\limplies,\liff,\lxor\}}\\
  \tr(\alpha\U\beta)&=\tr(\beta)\lor(\tr(\alpha)\!\land\!\tikz[baseline=(0.base)]\node[termn](0){$\alpha\U\beta$};) &
  \tr(\F \alpha)&=\tr(\alpha)\!\lor\!\tikz[baseline=(0.base)]\node[termn](0){$\F\alpha$};\\
\tr(\alpha\R\beta)&=\tr(\beta)\land(\tr(\alpha)\!\lor\!\tikz[baseline=(0.base)]\node[terma](0){$\alpha\R\beta$};)&
\tr(\G \alpha)&=\tr(\alpha)\!\land\!\tikz[baseline=(0.base)]\node[terma](0){$\G\alpha$};
\end{align*}}

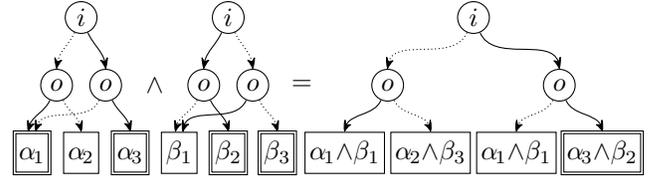
\begin{figure}[tbp]
  \centering
\begin{tikzpicture}[xscale=.815,yscale=.9]
  \node[inode] at (-1,-1) (i1) {$i$};
  \node[inode] at (-1.4,-2) (o1) {$o$};
  \node[inode] at (-.6,-2) (o2) {$o$};
  \node[terma] at (-1.8,-3) (t1) {$\!\alpha\kern-1pt_1\vphantom{\beta}\!$};
  \node[termn] at (-1,-3) (t2) {$\!\alpha\kern-1pt_2\vphantom{\beta}\!$};
  \node[terma] at (-.2,-3) (t3) {$\!\alpha\kern-1pt_3\vphantom{\beta}\!$};
  \draw[high] (i1) to[out=-65,in=90] (o2);
  \draw[low] (i1) to[out=-115,in=90] (o1);
  \draw[high] (o1) to[out=-115,in=90] (t1.100);
  \draw[low] (o1) to[out=-65,in=90] (t2.90);
  \draw[high] (o2) to[out=-65,in=90] (t3.90);
  \draw[low] (o2) to[out=-115,in=90] (t1.80);

  \node at (0.2,-2) {$\land$};

  \node[inode] at (1.4,-1) (i2) {$i$};
  \node[inode] at (1,-2) (o3) {$o$};
  \node[inode] at (1.8,-2) (o4) {$o$};
  \node[termn] at (.6,-3) (t4) {$\!\beta_1\!$};
  \node[terma] at (1.4,-3) (t5) {$\!\beta_2\!$};
  \node[terma] at (2.2,-3) (t6) {$\!\beta_3\!$};
  \draw[low] (i2) to[out=-65,in=90] (o4);
  \draw[high] (i2) to[out=-115,in=90] (o3);
  \draw[low] (o3) to[out=-115,in=90] (t4.100);
  \draw[high] (o3) to[out=-65,in=90] (t5.90);
  \draw[low] (o4) to[out=-65,in=90] (t6.90);
  \draw[high] (o4) to[out=-115,in=90] (t4.80);

  \node at (2.6,-2) {$=$};

  \node[inode] at (5.4,-1) (i3) {$i$};
  \node[inode] at (4,-2) (o5) {$o$};
  \node[inode] at (6.8,-2) (o6) {$o$};
  \node[termn] at (3.3,-3) (t7) {$\!\alpha\kern-1pt_1{\land}\beta_1\!$};
  \node[termn] at (4.7,-3) (t8) {$\!\alpha\kern-1pt_2{\land}\beta_3\!$};
  \node[termn] at (6.1,-3) (t9) {$\!\alpha\kern-1pt_1{\land}\beta_1\!$};
  \node[terma] at (7.5,-3) (t10) {$\!\alpha\kern-1pt_3{\land}\beta_2\!$};
  \draw[low] (i3) to[out=-115,in=90] (o5);
  \draw[high] (i3) to[out=-65,in=90] (o6);
  \draw[low] (o5) to[out=-65,in=90] (t8);
  \draw[high] (o5) to[out=-115,in=90] (t7);
  \draw[low] (o6) to[out=-115,in=90] (t9);
  \draw[high] (o6) to[out=-65,in=90] (t10);
\end{tikzpicture}
\caption[Conjunction of two MTBDDs]{Conjunction of two MTBDDs with \LTLf terminals.
  As typical in BDD-based representations,
    a node \begin{tikzpicture}[baseline=(i.base)]
      \node[inode] at (0,0) (i) {$x$};
      \node[overlay] at (0.75,0.15) (l) {$\ell$};
      \node[overlay] at (0.75,-0.15) (h) {$h$};
      \draw[low,overlay] (i) -> (l);
      \draw[high] (i) -> (h);
    \end{tikzpicture}~~~~should be read as ``if $x$ then $h$ else $\ell$''.\label{fig:MTBDDop}}
\end{figure}


Boolean operators that appear to the right of the equal sign are
applied to MTBDDs using algorithms similar to those of BDDs, with the
difference that $\LTLf$-labeled terminals are combined to form new
terminals:
$(\alpha,a)\odot(\beta,b) = ([\alpha\odot\beta]_\propequiv,a\odot b)$, as
done with formula progressions~(Section~\ref{sec:fp}).
Figure~\ref{fig:MTBDDop} shows an example of a conjunction of two
MTBDDs, that we shall reuse later.

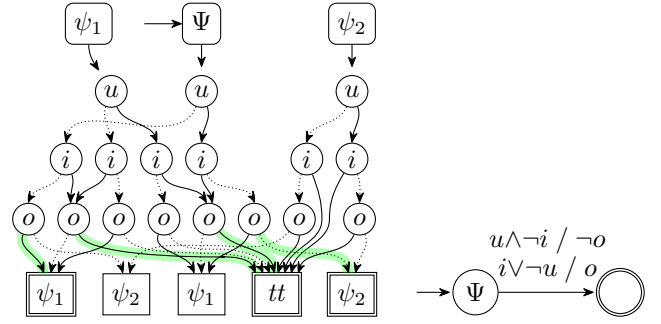
\begin{figure}[tbp]
  \centering
\begin{tikzpicture}
  \node[root] at (0.5,-0.4) (r1) {$\psi_1$};
  \node[root] at (2,-0.4) (r2) {$\Psi\vphantom{\psi_2}$};
  \node[left=5mm of r2] (init) {};
  \draw[->](init) -- (r2);
  \node[root] at (4,-0.4) (r3) {$\psi_2$};
  \node[inode] at (.8,-1.3) (u1) {$u$};
  \node[inode] at (2,-1.3) (u2) {$u$};
  \node[inode] at (4,-1.3) (u3) {$u$};
  \node[inode] at (.2,-2.2) (i1) {$i$};
  \node[inode] at (.8,-2.2) (i2) {$i$};
  \node[inode] at (1.4,-2.2) (i3) {$i$};
  \node[inode] at (2,-2.2) (i4) {$i$};
  \node[inode] at (3.4,-2.2) (i5) {$i$};
  \node[inode] at (4,-2.2) (i6) {$i$};
  \node[inode] at (-.3,-3) (o1) {$o$};
  \node[inode] at (.3,-3) (o2) {$o$};
  \node[inode] at (.9,-3) (o3) {$o$};
  \node[inode] at (1.5,-3) (o4) {$o$};
  \node[inode] at (2.1,-3) (o5) {$o$};
  \node[inode] at (2.7,-3) (o6) {$o$};
  \node[inode] at (3.3,-3) (o7) {$o$};
  \node[inode] at (4.1,-3) (o8) {$o$};
  \node[terma] at (0,-4) (t1) {$\psi_1$};
  \node[termn] at (1,-4) (t2) {$\psi_2$};
  \node[termn] at (2,-4) (t3) {$\psi_1$};
  \node[terma] at (3,-4) (t4) {$\,\ttrue\,\vphantom{\psi_2}$};
  \node[terma] at (4,-4) (t5) {$\psi_2$};

  \draw[rlink] (r1) to[out=-90,in=120] (u1);
  \draw[rlink] (r2) to[out=-90,in=90] (u2);
  \draw[rlink] (r3) to[out=-90,in=90] (u3);

  \draw[low] (u1) to[out=-110,in=90] (i2);
  \draw[high] (u1) to[out=-70,in=90] (i3);
  \draw[low] (u2) to[out=-110,in=90] (i1);
  \draw[high] (u2) to[out=-70,in=90] (i4);
  \draw[low] (u3) to[out=-110,in=90] (i5);
  \draw[high] (u3) to[out=-70,in=90] (i6);

  \draw[low] (i1) to[out=-110,in=90] (o1);
  \draw[high] (i1) to[out=-70,in=100] (o2);
  \draw[low] (i2) to[out=-70,in=90] (o3);
  \draw[high] (i2) to[out=-110,in=80] (o2);
  \draw[low] (i3) to[out=-90,in=90] (o4);
  \draw[high] (i3) to[out=-70,in=100] (o5);
  \draw[low] (i4) to[out=-70,in=90] (o6);
  \draw[high] (i4) to[out=-90,in=80] (o5);
  \draw[low] (i5) to[out=-110,in=90] (o7);
  \draw[high] (i5) to[out=-70,in=90] ($(o7)+(3mm,0mm)$) to [out=-90,in=90] (t4.69);
  \draw[low] (i6) to[out=-70,in=90] (o8);
  \draw[high] (i6) to[out=-120,in=90] ($(o7)+(4mm,0mm)$) to [out=-90,in=90] (t4.55);

  \draw[wine] (o1) to[out=-110,in=90] (t1.110);
  \draw[wine] (o2) to[out=-70,in=0,looseness=.5] ($(o4)!.5!(t4)$) to[out=0,in=90,looseness=1] (t4.136);
  \draw[wine] (o5) to[out=-50,in=90] (t4.97);
  \draw[wine] (o6) to[out=-70,in=90] (t5.100);

  \draw[low] (o1) to[out=-70,in=90] (t2.100);
  \draw[high] (o1) to[out=-110,in=90] (t1.110);
  \draw[low] (o2) to[out=-110,in=90] (t1.90);
  \draw[high] (o2) to[out=-70,in=0,looseness=.5] ($(o4)!.5!(t4)$) to[out=0,in=90,looseness=1] (t4.136);
  \draw[low] (o3) to[out=-70,in=90,looseness=.65] (t4.124);
  \draw[high] (o3) to[out=-110,in=90] (t1.70);
  \draw[low] (o4) to[out=-70,in=180,looseness=.8] ($(o4)!.5!(t4) + (0,1mm)$) to[out=0,in=90,looseness=.8] (t4.110);
  \draw[high] (o4) to[out=-110,in=90] (t3.110);
  \draw[low] (o5) to[out=-110,in=90] (t3.90);
  \draw[high] (o5) to[out=-50,in=90] (t4.97);
  \draw[low] (o6) to[out=-70,in=90] (t5.100);
  \draw[high] (o6) to[out=-110,in=90] (t3.70);
  \draw[low] (o7) to[out=-130,in=0,looseness=.9] ($(o5)+(0,-3.5mm)$) to[out=-180,in=90] (t2.80);
  \draw[high] (o7) to[out=-90,in=90] (t4.86);
  \draw[low] (o8) to[out=-70,in=90] (t5.80);
  \draw[high] (o8) to[out=-110,in=90] (t4.44);
\end{tikzpicture}
\hfill
\begin{tikzpicture}[automaton,every state/.style={minimum size=6mm,fill=white,inner sep=2pt},
                    every initial by arrow/.append style={overlay}]
    \node[state,initial] (0) {$\Psi$};
    \node[state,accepting,right=13mm of 0] (1) {};
    \path[->] (0) edge node[above,align=center] {$u{\land}\lnot i~/~\lnot o$\\$i{\lor}\lnot u~/~o$}  (1);
  \end{tikzpicture}
\caption[MTDFA]{(left) An MTDFA for the formula $\Psi=\psi_1\land\psi_2$ where $\psi_1=(\G\F u) \limplies \F(i\liff o)$ and $\psi_2=(\G\F \lnot u) \limplies \F(i \lor o)$.  A root \begin{tikzpicture}[baseline=(r.base)]
      \node[root] at (0,0) (r) {$\alpha$};
      \node[inner sep=0] at (.8,0) (b) {$m$};
      \draw[rlink] (r) -> (b);
    \end{tikzpicture} indicates that $\tr(\alpha)=m$.
    If this automaton is interpreted as a game where the environment plays $\CI=\{u,i\}$,
    and the controller plays $\CO=\{o\}$, then the controller has a strategy to reach
    an accepting terminal if it takes the highlighted edges.  (right) The corresponding controller
    as a terminating Mealy machine.
    \label{fig:mtdfa}}
\end{figure}

The different paths in the MTBDD computed by $\tr(\varphi)$ represent
all possible formula progressions for $\varphi$, but the above
efficiently constructs all progressions at once, without having to compute
$\fp{\varphi, w}$ for each $w \in \BB^\CP$ independently.  By
computing such an MTBDD for all $\LTLf$ terminals reachable from
$\varphi$, one can obtain an MTDFA representation as shown in
Figure~\ref{fig:mtdfa} for the formula $\Psi$ of Example~\ref{ex:psi}.

\medskip
Assuming an appropriate ordering of the MTBDD variables $\CP$ (e.g.,
input variables before output variables for Mealy semantics), the
constructed MTDFA can be interpreted directly as a two-player game to
solve the corresponding $\LTLf$ synthesis problem under full
observability~($\CI=\{i,u\}$ and $\CO=\{o\}$).  In this construction,
accepting terminals (such as
\tikz[baseline=(X.base)]\node[terma](X){$\ttrue$}; in
Fig.~\ref{fig:mtdfa}) are the targets of the reachability game, so
their successors do not need to be computed.  Finally, the fact that
$\tr(\cdot)$ computes all successors of one state enables an
on-the-fly construction of the MTDFA while solving the
game~\cite[Algorithm~2]{duret.25.ciaa}.  For instance in
Figure~\ref{fig:mtdfa}, the highlighted winning strategy can already
be found after computing $\tr(\Psi)$, so the algorithm can stop
without exploring the discovered successors.

%

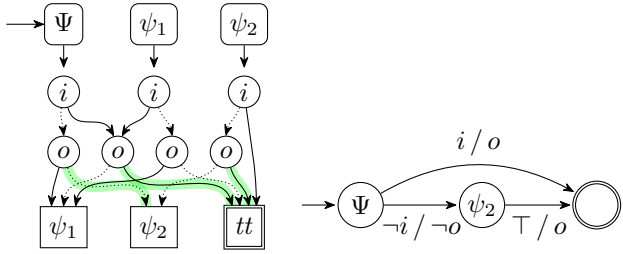
\begin{figure}[tbp]
  \centering
\begin{tikzpicture}[xscale=.6]
  \node[root] at (0,-.3) (r1) {$\Psi\vphantom{\psi_2}$};
  \node[left=5mm of r1] (init) {};
  \draw[->](init) -- (r1);
  \node[root] at (2,-.3) (r2) {$\psi_1$};
  \node[root] at (4,-.3) (r3) {$\psi_2$};
  \node[inode] at (0,-1.2) (i1) {$i$};
  \node[inode] at (2,-1.2) (i2) {$i$};
  \node[inode] at (4,-1.2) (i3) {$i$};

  \draw[rlink] (r1) to[out=-90,in=90] (i1);
  \draw[rlink] (r2) to[out=-90,in=90] (i2);
  \draw[rlink] (r3) to[out=-90,in=90] (i3);

  \node[inode] at (0,-2) (o1) {$o$};
  \node[inode] at (1.2,-2) (o2) {$o$};
  \node[inode] at (2.4,-2) (o3) {$o$};
  \node[inode] at (3.6,-2) (o4) {$o$};

  \node[termn] at (0,-3) (t1) {$\psi_1$};
  \node[termn] at (2,-3) (t2) {$\psi_2$};
  \node[terma] at (4,-3) (t3) {$\ttrue\vphantom{\psi_2}$};

  \draw[wine] (o1) to[out=-70,in=90,looseness=.9] (t2.110);
  \draw[wine] (o2) to[out=-60,in=100,looseness=.7] (t3.120);
  \draw[wine] (o4) to[out=-70,in=90] (t3.78);

  \draw[low] (i1) to[out=-100,in=90] (o1);
  \draw[high] (i1) to[out=-70,in=110] (o2);
  \draw[low] (i2) to[out=-70,in=90] (o3);
  \draw[high] (i2) to[out=-110,in=70] (o2);
  \draw[low] (i3) to[out=-110,in=90] (o4);
  \draw[high] (i3) to[out=-80,in=90] (t3.60);

  \draw[low] (o1) to[out=-70,in=90,looseness=.9] (t2.110);
  \draw[high] (o1) to[out=-110,in=90] (t1.120);

  \draw[low] (o2) to[out=-130,in=90] (t1.90);
  \draw[high] (o2) to[out=-60,in=100,looseness=.7] (t3.120);

  \draw[low] (o3) to[out=-45,in=90] (t3.102);
  \draw[high] (o3) to[out=-130,in=90,looseness=.7] (t1.60);

  \draw[low] (o4) to[out=-135,in=90] (t2.70);
  \draw[high] (o4) to[out=-70,in=90] (t3.78);
\end{tikzpicture}
  \hfill
  \begin{tikzpicture}[automaton,every state/.style={minimum size=6mm,fill=white,inner sep=0pt}]
    \node[state,initial] (0) {$\Psi$};
    \node[state,right=10mm of 0] (1) {$\psi_2$};
    \node[state,accepting,right=9mm of 1] (2) {};
    \path[->] (0) edge[bend left] node[above] {$i\,/\,o$} (2);
    \path[->] (0) edge node[below] {$\lnot i\,/\, \lnot o$} (1);
    \path[->] (1) edge node[below] {$\top\,/\, o$} (2);
  \end{tikzpicture}
\caption[Belief-state MTDFA]{(left) Belief-state MTDFA for $\Psi$, assuming that $\CU=\{u\}$ is unobservable.
  If this automaton is interpreted as a game where the controller plays $\CO=\{o\}$,
  the highlighted edges represent one possible winning strategy. (right) The corresponding controller as a terminating Mealy machine.
  \label{fig:mtdfa-quantified}}
\end{figure}


\medskip
\noindent\textbf{MTDFA-based Observable Progression.}
The partial-observability extension introduces support for partial observability in the above architecture.

To do so, we use an MTBDD operation to universally quantify the
unobservable variables $\CU\subseteq\CP$ in the MTBDD computed by
$\tr(\varphi)$ during on-the-fly exploration of the MTDFA.  If
$\tr(\varphi)$ was a symbolic representation of the set of
progressions of the form $\fp{\varphi, w}$ for $w \in \BB^\CP$,
then $\forall\CU.\tr(\varphi)$ is the symbolic representation
of the set of all observable progressions
$\fpObs{\varphi, w_o}$ for $w_o \in \BB^{\CP_o}$.

The MTBDD $\forall\CU.\tr(\varphi)$ is computed bottom-up using an
algorithm similar to those of BDDs: an internal MTBDD node labeled
with a variable in $\CU$ is replaced by the conjunction of its
children.  For instance, the MTBDD computed for $\tr(\Psi)$ in
Figure~\ref{fig:mtdfa} has its top-level node labeled by $u\in\CU$, so
it should be replaced by the conjunction of its two children.  This is
exactly the operation depicted in Figure~\ref{fig:MTBDDop}, with
$\alpha_1=\alpha_2=\psi_1$, $\beta_1=\beta_3=\psi_2$, and
$\alpha_3=\beta_2=\ttrue$.

Doing this quantification for every state reached during the translation
gives the belief-state MTDFA of Figure~\ref{fig:mtdfa-quantified}.
To make this quantification efficient, we always put all unobservable
variables $\CU$ at the bottom of the variable order: thus all nodes
that should be removed appear in subtrees of the MTBDDs, and their
quantification replaces those subtrees by terminals labeled by the conjunction of their leaves.

\begin{theorem}\label{th:translation}
For $\varphi\in\LTLf(\CP)$, with $\CP=\CP_o\uplus\CU$, let
$\CA^B_\varphi=\langle \CS,\CP_o,\iota^B,\Delta^B\rangle$ be the belief-state MTDFA obtained by setting $\iota^B=[\varphi]_{\propequiv}$, $\Delta^B(s)=\forall\CU.\tr(s)$, and letting $\CS$ be the smallest subset of $\LTLf(\CP)$ such that (i) $\iota^B\in\CS$, and (ii) for every $s\in\CS$ and  every terminal labeled $(\alpha,b)$ in $\Delta^B(s)$,  $\alpha\in\CS$.  Then $\CS$ is finite and $\CA^B_\varphi$ is a symbolic representation of the belief-state DFA defined in Definition~\ref{def:ltlf2bdfa}.
\end{theorem}
\begin{proof}
  Finiteness follows directly from Theorem~\ref{thm:bdfaSize}: $\tr(s)$ contains only Boolean combinations of subformulas of $\varphi$, and the quantification $\forall\CU$ only combines leaves by conjunction. Hence all states are still Boolean combinations of subformulas of $\varphi$, up to propositional equivalence, and there are finitely many of them.

  Correctness follows from the definition of MTBDD universal quantification. For every $s\in\CS$ and $w_o\in\BB^{\CP_o}$, evaluating $\forall\CU.\tr(s)$ under $w_o$ yields the conjunction of the progressions obtained from all completions of $w_o$ with valuations of the unobservable variables: $(\forall\CU.\tr(s))(w_o) = \bigwedge\bigl\{\,\fp{\varphi,w}\ \bigm|\ w\in\BB^{\CP},\ w|_{\CP_o}=w_o\,\bigr\}.$ By definition, this is exactly $\fpObs{s,w_o}$. Therefore, $\Delta^B$ induces the same successors and accepting transitions as the belief-state DFA of Definition~\ref{def:ltlf2bdfa}.
\end{proof}

\noindent\emph{On-the-fly synthesis} is enabled by solving the induced reachability game during the exploration of the belief-state MTDFA, without constructing the full automaton upfront.   On the
example of Figure~\ref{fig:mtdfa-quantified}, an on-the-fly
construction that explores $\psi_2$ before $\psi_1$ would solve the game
before exploring $\psi_1$.  The game-solving procedure was not
changed; it was only applied to the belief states.

\section{Empirical Evaluation}\label{sec:evaluation}


In this section, we evaluate our on-the-fly synthesis approach by
comparing our implementation against state-of-the-art tools for \LTLf
synthesis under partial observability. 
As mentioned in Section~\ref{sec:impl}, our implementation is available in \texttt{ltlfsynt}, starting from Spot~2.15. In the experimental plots and tables, we denote that partial-observability configuration by \spototf.
In addition, to evaluate the practical benefits of integrating the on-the-fly game solving~(see Sec~\ref{sec:otf-game}), we implemented a variant of \spototf, named \spotbsc, in which on-the-fly game solving is disabled. In this configuration, a belief-state MTDFA is first constructed in full, and only afterward a reachability game is solved on the resulting automaton.
All of our experiments were run on a machine with an Intel Xeon Gold 6130 CPU with 2.10GHz, 64 cores and 188.5 GB of RAM. We limited each test to 90 seconds.

\subsection{Baseline Tools}\label{sec:expsyft}


The tools used as baselines in our evaluation are those introduced by~\citet{TV2020}, which implement different approaches.
All tools were built on the symbolic synthesis framework \emph{Syft} for \LTLf under full observability~\cite{zhu.17.ijcai}.
The first approach is belief-state based, implemented in \syftbsc. This approach is similar to belief constructions commonly used in planning~\cite{BonetGeffner2000,HoffmannBrafman2005}, which first builds a deterministic finite automaton~(DFA) for the \LTLf specification and then applies a second subset construction to turn it into a belief-state DFA, with the unobservable variables omitted from the transition function. The second approach is projection based, and implemented in \syftp. It instead starts from a nondeterministic finite automaton~(NFA) and applies a sequence of manipulations to obtain an equivalent belief-state DFA. The third approach is MSO based, implemented in \syftmso. It translates the \LTLf specification into a monadic second-order logic~(MSO) formula with the unobservable variables universally quantified, and then constructs a DFA directly for this MSO formula, which is equivalent to the belief-state DFA obtained by the other two approaches. Since \syftp has been shown to perform worse than \syftbsc and \syftmso in their experimental evaluation, we exclude it from our comparison. Refer to \citet{TV2020} for further details on these approaches.

\subsection{Benchmarks}\label{sec:expbench}
We ran our experiments on several benchmark sets, which we divide into two categories. The first category is the \lucasbench, introduced by~\citet{TV2020} and designed with partial-observability in mind. That is, every instance in this benchmark already contains a set of unobservable environment variables. This category includes \emph{Coin game} (8 instances), \emph{Traveling Target} (9 instances), and \emph{Private Peek}, which consists of 16 sub-categories with 100 instances each~(1600 instances in total).

The second category is {\syntcompbench}, taken from the annual reactive synthesis competition~(SYNTCOMP)~\cite{syntcomp-benchmark-url}.
These benchmarks do not originally contain any unobservable variables. As such, we randomly selected half of the environment variables and made them unobservable. To get a better chance that the unobservable variables affect the overall result, we selected benchmarks from designated two-player games, in which such variables are highly likely to affect the outcome.
In total, we have 150 instances: \emph{Single Counter}~(20 instances), \emph{Double Counter}~(20 instances), \emph{Chomp Game}~(22 instances), and \emph{Nim Game}~(88 instances). We ran each instance 10 times, and in each run randomly selected a fresh set of 50\% environment variables considered unobservable, without repeating any previously chosen set, to reduce the impact of randomness and obtain more robust and representative performance measurements.

\noindent\textbf{Correctness.} The implementation of \spototf and \spotbsc was verified by comparing the results (realizable/unrealizable) returned by them with those from \syftbsc and \syftmso. The results were consistent among all tools in all solved cases.

\subsection{Evaluation Results}\label{sec:expsetting}
\ifappendix
Due to space limitations, we only present overall comparison results here. A detailed comparison on individual benchmark families can be found in the appendix.
\fi

Figure~\ref{fig:po-cactus} is a cactus plot for \lucasbench, comparing the end-to-end runtime. All benchmark instances are sorted by runtime for each tool. The colors correspond to \spototf (green), \spotbsc (yellow), \syftbsc (blue) and \syftmso (red). Recall that in this category the partial observability is native to the specifications.
%
\begin{figure}[tb]
  \centering \includegraphics[width=0.45\textwidth]{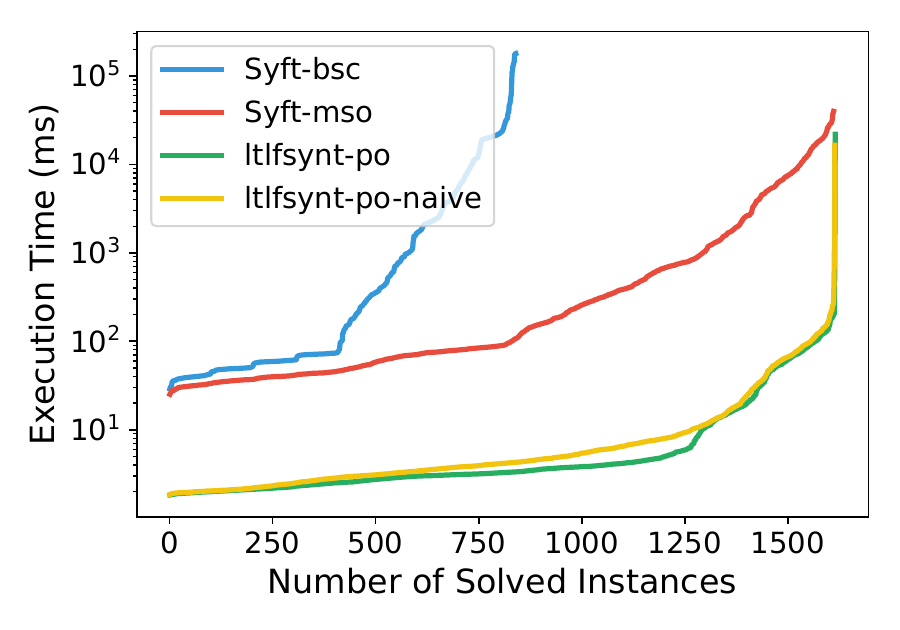}
  \caption{Cactus plot for \lucasbench\label{fig:po-cactus}}
   \includegraphics[width=0.45\textwidth]{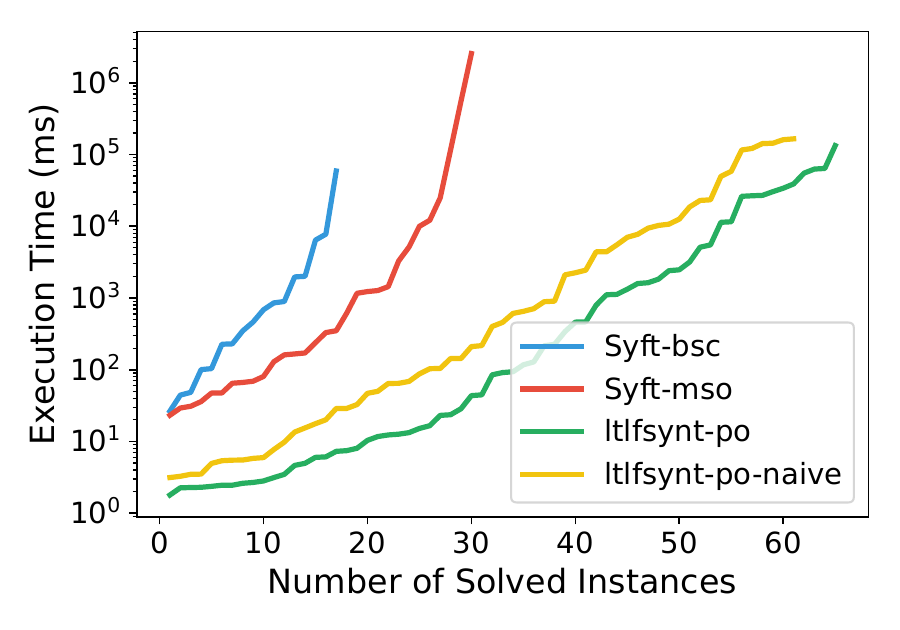}
  \caption{Cactus plot for \syntcompbench\label{fig:syntcomp-cactus}}
\end{figure}
%
%
%
This figure shows a significant performance improvement for \spototf over \syftbsc and \syftmso, with the speedup ranging from 27.65x to 126.6x in average runtime across the three different benchmark families. The speedup is measured by taking the total running time of both tools (\spototf and the best-performing between \syftbsc and \syftmso) for each individual benchmark family and computing their ratio.

Similarly, Figure~\ref{fig:syntcomp-cactus} is a cactus plot for the \syntcompbench, in which the unobservable variables are randomly selected, following the same palette of colors as in Figure~\ref{fig:po-cactus}. Here too, the plot shows a significant performance improvement for \spototf over \syftbsc and \syftmso, with the speedup ranging from 4.81x to 7,892.82x in average runtime across the different benchmark families. 


Overall, \spototf dramatically outperforms current state-of-the-art tools. These results demonstrate the strength of both our theoretical on-the-fly approaches as described in Section~\ref{sec:otf}, and our practical MTDFA-based implementation, as described in Section~\ref{sec:impl}.

We conducted a more detailed comparison to better understand the contribution of the on-the-fly game solving component in \spototf. In general, \spototf consistently outperforms \spotbsc across all benchmarks instances.
%
\begin{table}[tb]
  \centering
  \caption{Combined average speedups of \spototf compared to \spotbsc across all benchmarks.}
  \label{tab:combined-speedup}
  \begin{tabularx}{0.48\textwidth}{Xlr}
    \toprule
    \textbf{Source} & \textbf{Family} & \textbf{Speedup} \\
    \midrule
    \lucasbench    & Coin-Game       & 16.54$\times$ \\
                   & Moving-Target   & 1.22$\times$  \\
                   & Private-Peek    & 1.17$\times$  \\
    \midrule
   \emph{SYNTCOMP-fin Bench.} & Chomp           & 4.62$\times$  \\
                   & Counter         & 4.46$\times$  \\
                   & Counters-Double & 5.11$\times$  \\
                   & Nim             & 4.61$\times$  \\
    \bottomrule
  \end{tabularx}
\end{table}
%
%

Table~\ref{tab:combined-speedup} summarizes the speedup on both benchmark categories. Specifically, for the \lucasbench, we can see that \spototf achieves a speedup ranging from 1.17x to 16.54x over \spotbsc across the individual benchmark families. For the \syntcompbench, the speedups range from 4.46x to 5.11x.
These results indicate that the major improvement in performance is due to our belief-state MTDFA construction techniques. At the same time, they also show that the integration of on-the-fly game solving plays an important role. By solving the reachability game already during the construction of the belief-state MTDFA, \spototf can avoid exploring parts of the automaton that are not needed, which further improves the overall performance.

\section{Conclusion and Future Work}\label{sec:conclusion}

In this paper, we presented an on-the-fly approach to \LTLf synthesis under partial observability based on observable progression. This approach incrementally constructs the belief-state DFA, universally quantifying unobservable environment variables during construction, enabling a tight integration of belief-state DFA construction and game solving. The resulting implementation, leveraging MTDFA representation, shows significant improvements compared to existing approaches. More generally, our results show that MTDFA-based on-the-fly techniques are not limited to the fully-observable setting.
%
Note that our belief-state DFA construction through observable progression essentially builds an automaton incrementally that is equivalent to a formula of the form $\forall \CU. \varphi$, with $\varphi \in \LTLf(\CP)$ and $\CU \subseteq \CP$. This technique can be adapted to support DFA construction for $\exists \CU. \varphi$ by changing universal quantification to existential quantification on the MTBDDs, and more generally for quantified \LTLf formulas~(\QLTLf). Thus our approach can be extended to settings that apply synthesis
to \QLTLf formulas, such as in~\cite{hagemeier2025ltlf}.

\section*{Acknowledgments}
We thank the reviewers for their insightful comments.
SZ and SC are partly supported by the Royal Society International Exchanges grant IES\textbackslash{}R2\textbackslash{}252189. NA is supported by the Open University of Israel MSc Excellence Scholarship and by The Open University of Israel Research Fund 47372.

\section*{AI Declaration}
Generative AI was used for language polishing and experimental script development. All scientific results, proofs, artifacts, and conclusions were independently verified by the authors, who assume full responsibility for the final content of the paper.

\bibliographystyle{kr}
\bibliography{biblio}

\ifappendix
\clearpage
\appendix
\section{Supplementary Evaluation}\label{sec:appEval}
\subsection{Full Benchmark Results}\label{sec:appResults}

We provide detailed results of our experiments. We first discuss the \lucasbench, in which the partial observability is originated in the specification. Then we discuss the \syntcompbench, in which the unobservable variables were randomly selected.

\paragraph{\lucasbench}
Figures~\ref{fig:po-perf-1},~\ref{fig:po-perf-2}, and~\ref{fig:po-perf-3} show the performance of \spototf, \spotbsc, \syftbsc and \syftmso on the \lucasbench. 
\begin{figure}[htp]
  \centering
  \includegraphics[width=0.45\textwidth]{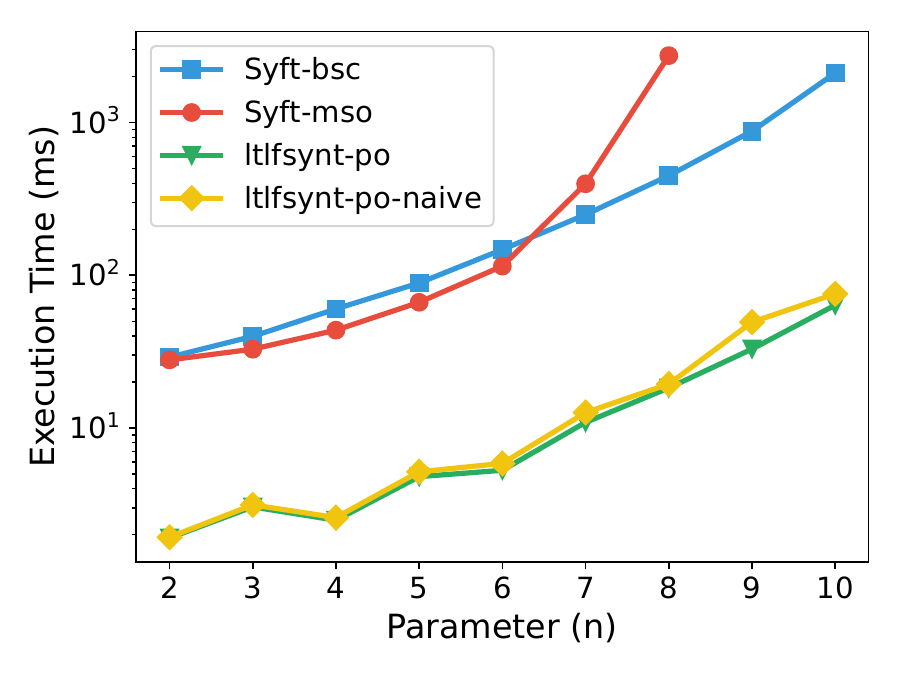}\par\vspace{0.5em}
  \caption{Moving Target Performance}
  \label{fig:po-perf-1}
\end{figure}

\begin{figure}[htp]
  \centering
  \includegraphics[width=0.45\textwidth]{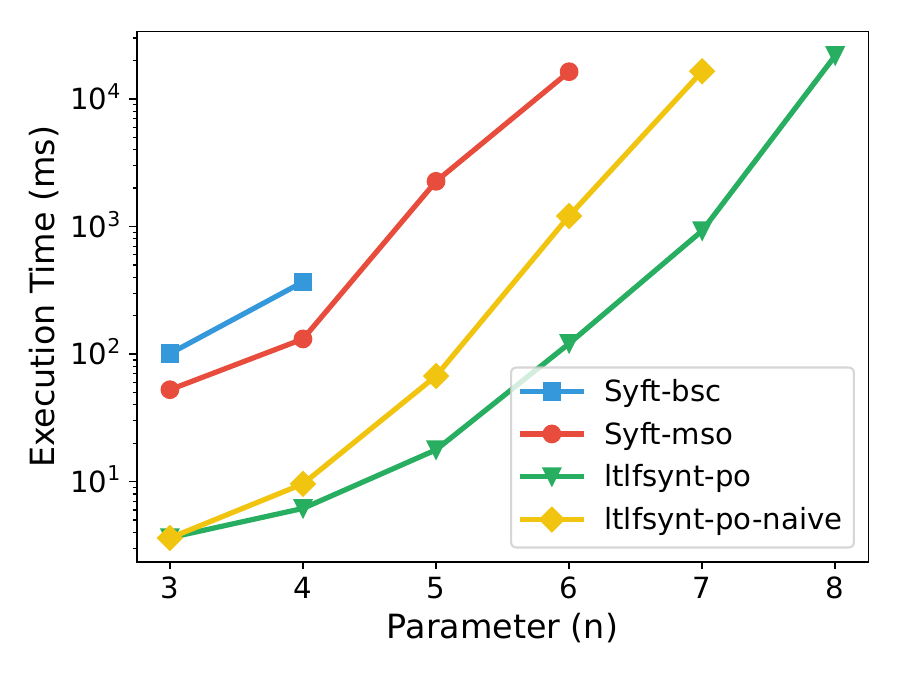}\par\vspace{0.5em}
  \caption{Coin Game Performance}
  \label{fig:po-perf-2}
\end{figure}

\begin{figure}[htp]
  \centering
  \includegraphics[width=0.45\textwidth]{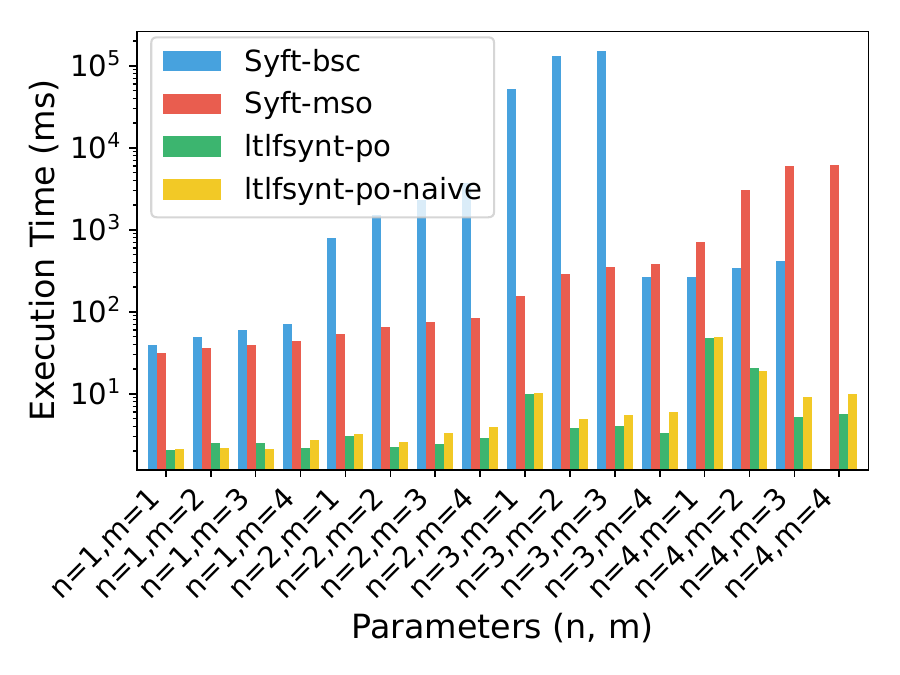}
  \caption{Private Peek Performance}
  \label{fig:po-perf-3}
\end{figure}
\FloatBarrier

Table~\ref{tab:po-speedup} shows the average speedup in each benchmark family. 
The speedup is measured by taking the total running time of both tools~(\spototf and the best-performing one between \syftbsc and \syftmso) for each individual benchmark family and computing their ratio.

\begin{table}[htp]
  \centering
  \begin{tabular}{lr}
    \toprule
    Benchmark Family & Speedup \\
    \midrule
    Coin-Game        & 126.60$\times$ \\
    Moving-Target    & 27.65$\times$ \\
    Private-Peek     & 126.00$\times$ \\
    \bottomrule
  \end{tabular}
  \caption{Average speedups of \spototf compared to the best-performing Syft implementation (\syftbsc or \syftmso) across \lucasbench.}
  \label{tab:po-speedup}
\end{table}


\paragraph{\syntcompbench}

Similarly, Figures~\ref{fig:perf-chomp},~\ref{fig:perf-nim},~\ref{fig:perf-counter} and~\ref{fig:perf-counters-double} show the performance of \spototf, \spotbsc, \syftbsc and \syftmso on the \syntcompbench.


  \begin{figure}[htp]
    \centering
    \includegraphics[width=0.45\textwidth]{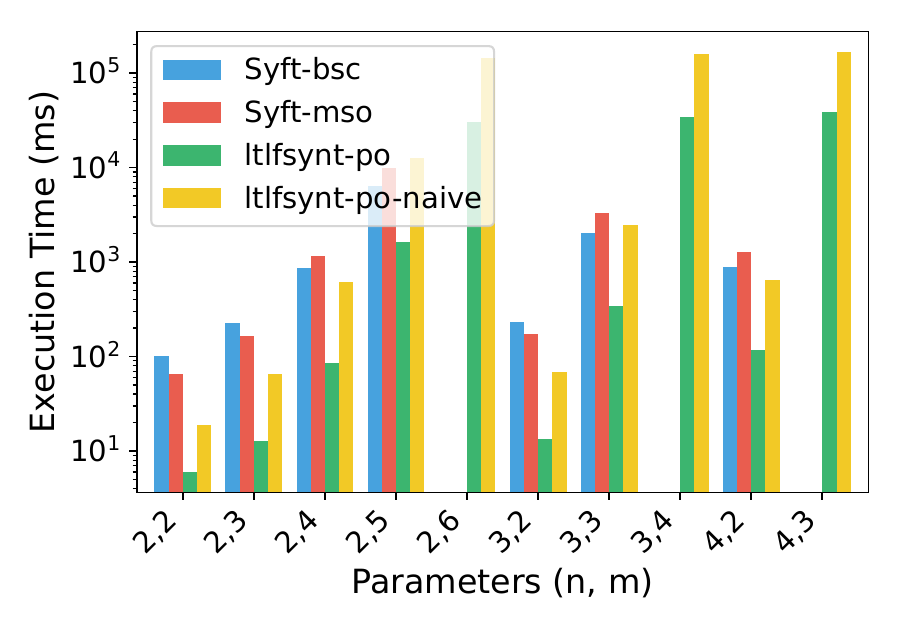}
    \caption{Chomp Performance}
    \label{fig:perf-chomp}
  \end{figure}

  \begin{figure}[htp]
    \centering
    \includegraphics[width=0.45\textwidth]{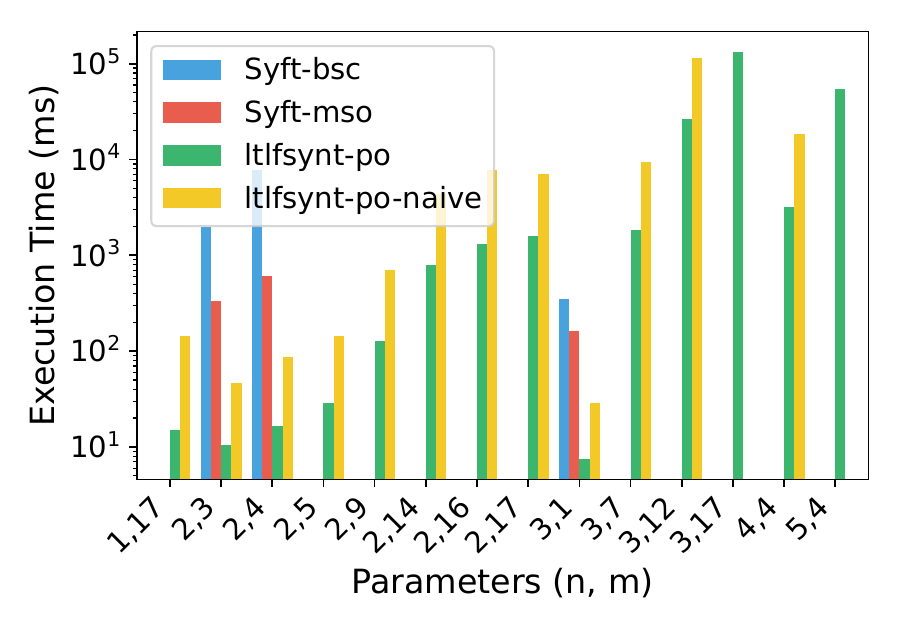}
    \caption{Nim Performance}
    \label{fig:perf-nim}
  \end{figure}

  \begin{figure}[htp]
    \centering
    \includegraphics[width=0.45\textwidth]{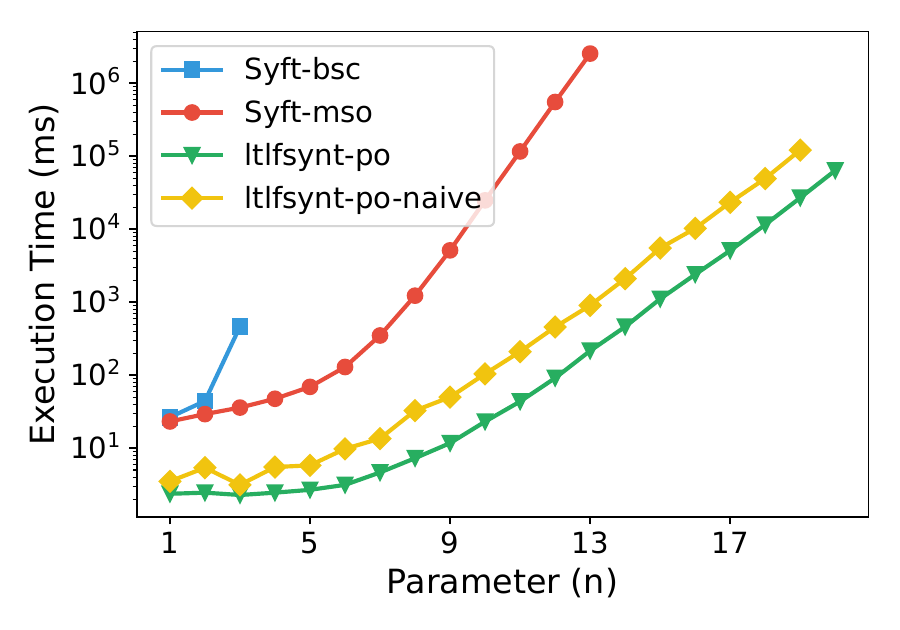}
    \caption{Counter Performance}
    \label{fig:perf-counter}
  \end{figure}

  \begin{figure}[htp]
    \centering
    \includegraphics[width=0.45\textwidth]{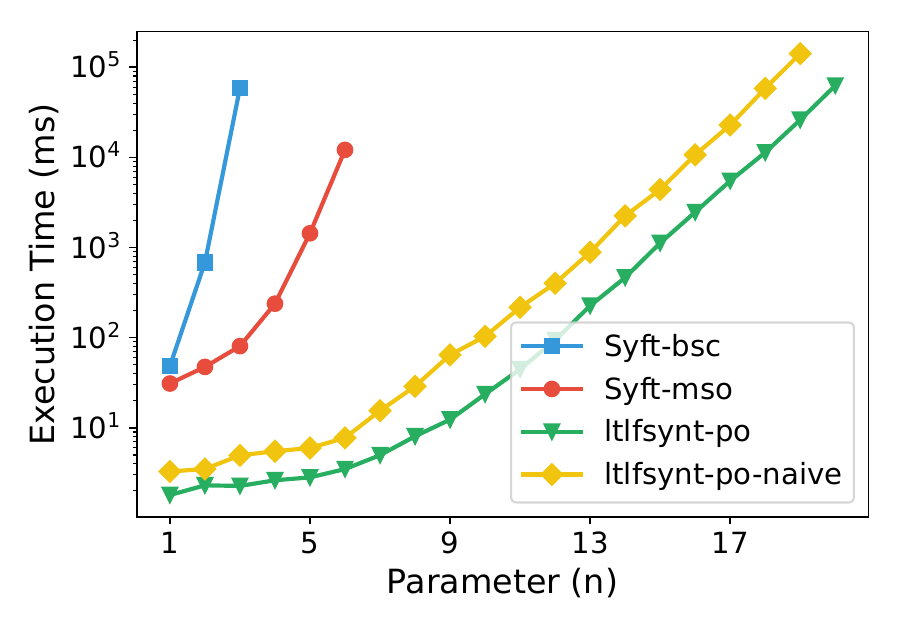}
    \caption{Double Counters Performance}
    \label{fig:perf-counters-double}
  \end{figure}


Table~\ref{tab:syntcomp-speedup} shows a speedup ranging from 4.81x to 7892.82x. It should be noted that both \spototf and \spotbsc solve a lot of benchmarks that the state-of-the-art tools could not solve due to either time constraints or memory constraints.

\begin{table}[htp]
  \centering
  \begin{tabular}{lr}
    \toprule
    Benchmark Family & Speedup \\
    \midrule
    Chomp           & 4.81$\times$ \\
    Counter         & 7,892.82$\times$ \\
    Counters-Double  & 915.84$\times$ \\
    Nim             & 32.10$\times$ \\
    \bottomrule
  \end{tabular}
  \caption{Average speedups of \spototf compared to the best-performing Syft implementation across \syntcompbench (50\% variables hidden). }
  \label{tab:syntcomp-speedup}
\end{table}

\FloatBarrier






\section{Reproducibility}

We describe how to reproduce the results that appear in this paper.

\subsection{Benchmarks}
\paragraph{\lucasbench} can be found online at \url{https://github.com/lucasmt/ltlf-po-benchmarks}
\paragraph{\syntcompbench} can be found online at \url{https://github.com/SYNTCOMP/benchmarks/tree/master/tlsf-fin}
\subsection{Our Tool}
The implementation is publicly available in Spot starting from version~2.15~(available at \url{https://spot.lre.epita.fr/install.html}). After installing Spot, the tool can be invoked as follows:

\paragraph{\spototf}

\texttt{./ltlfsynt} \texttt{<formula>}\\
\texttt{--ins=<inputs>}\\
\texttt{--outs=<outputs>}\\
\texttt{--unobservable-ins=<unobservables>}\\
\texttt{--semantics=<mealy or moore>}


\paragraph{\spotbsc}

\texttt{./ltlfsynt} \texttt{<formula>}\\
\texttt{--ins=<inputs>}\\
\texttt{--outs=<outputs>}\\
\texttt{--unobservable-ins=<unobservables>}\\
\texttt{--semantics=<mealy or moore>} \\
\texttt{--translation=restricted}

\subsection{Baseline Tools}

The version of Syft with support for partial observability can be found online at \url{https://github.com/lucasmt/Syft/tree/double-negation-dfa}. To install, follow the instructions in \url{https://github.com/lucasmt/Syft/blob/double-negation-dfa/INSTALL}. You might need to change the \texttt{CUDD\_ROOT} variable in \url{https://github.com/lucasmt/Syft/blob/double-negation-dfa/CMakeModules/Findcudd.cmake} to point to your local CUDD installation. Afterwards, add the \texttt{build/bin} folder to your path in order to have access to the \texttt{ltlf2fol} and \texttt{Syft} binaries.

\paragraph{\syftbsc} To run \syftbsc on a benchmark \texttt{bm} composed of \LTLf file \texttt{bm.ltlf} and variable partition file \texttt{bm.part}, follow these steps:

\begin{enumerate}
    \item Run \texttt{ltlf2fol} to convert the \LTLf file into a MONA input file \texttt{bm.mona}:

    \texttt{ltlf2fol NNF bm.ltlf > bm.mona}

    \item Run MONA to generate a DFA file:

    \texttt{mona -u -xw bm.mona > bm.dfa}

    \item Run Syft on the resulting DFA, with the system as the starting player (0) and partial observability (implemented via symbolic belief-states construction):

    \texttt{Syft bm.dfa bm.part 0 partial dfa}
\end{enumerate}

\paragraph{\syftmso} To run \syftmso on a benchmark \texttt{bm} composed of \LTLf file \texttt{bm.ltlf} and variable partition file \texttt{bm.part}, follow these steps:

\begin{enumerate}
    \item Run \texttt{ltlf2fol} to convert the \LTLf file into a MONA input file \texttt{bm.mona}:

    \texttt{ltlf2fol NNF bm.ltlf > bm.mona}

    \item Quantify the unobservable variables universally in the MONA file (see the script in \url{https://github.com/lucasmt/ltlf-po-benchmarks/blob/master/quantify.py} for an example), producing a new file \texttt{bm.mona.quant}.

    \item Remove the unobservable variables from the variable partition file, producing a new file \texttt{bm.part.quant}.

    \item Run MONA to generate a DFA file from the universally-quantified MSO formula:

    \texttt{mona -u -xw bm.mona.quant > bm.dfa.quant}

    \item Run Syft on the resulting DFA, with the system as the starting player (0) and full observability (since partial observability has already been handled by universal quantification):

    \texttt{Syft bm.dfa.quant bm.part.quant 0 full dfa}
\end{enumerate}

Note that for the benchmarks in \url{https://github.com/lucasmt/ltlf-po-benchmarks}, the \texttt{.mona}, \texttt{.part}, \texttt{.mona.quant} and \texttt{.part.quant} files have already been generated, and therefore only the last two steps are necessary in each case.

\fi

\end{document}